\documentclass[12pt]{article}

\usepackage{latexsym,amsmath,amscd,amssymb,amsthm,graphics}
\usepackage{enumerate}
\usepackage[margin=2cm]{geometry}
\usepackage{graphicx}

\usepackage[square,authoryear]{natbib}
\usepackage[colorlinks]{hyperref}
\usepackage{url,color}
\usepackage{framed}

\usepackage[all]{xy}

\makeatletter

\@addtoreset{figure}{section}
\def\thefigure{\thesection.\@arabic\c@figure}
\def\fps@figure{h, t}
\@addtoreset{table}{bsection}
\def\thetable{\thesection.\@arabic\c@table}
\def\fps@table{h, t}
\@addtoreset{equation}{section}

\makeatother

\begin{document}

\newtheorem{theorem}{Theorem}[section]
\newtheorem{definition}[theorem]{Definition}
\newtheorem{lemma}[theorem]{Lemma}
\newtheorem{remark}[theorem]{Remark}
\newtheorem{proposition}[theorem]{Proposition}
\newtheorem{corollary}[theorem]{Corollary}
\newtheorem{example}[theorem]{Example}

\def\balpha{\boldsymbol\alpha}
\def\bom{\boldsymbol\omega}
\def\bOm{\boldsymbol\Omega}
\def\bPi{\boldsymbol\Pi}
\def\below#1#2{\mathrel{\mathop{#1}\limits_{#2}}}



\title{A geometric theory of selective decay in fluids with advected quantities}

\author{
Fran\c{c}ois Gay-Balmaz$^{1}$ and Darryl D. Holm$^{2}$
}
\addtocounter{footnote}{1}
\footnotetext{Laboratoire de M\'et\'eorologie Dynamique, \'Ecole Normale Sup\'erieure/CNRS, Paris, France. 
\texttt{gaybalma@lmd.ens.fr}
\addtocounter{footnote}{1} }
\footnotetext{Department of Mathematics, Imperial College, London SW7 2AZ, UK. 
\texttt{d.holm@ic.ac.uk}
\addtocounter{footnote}{1}}

\date{ 
\\Hamiltonian and Lagrangian mechanics, 45.20.Jj, 47.10.Df; 
\\Fluids, mathematical formulations, 47.10.A- }

\maketitle

\makeatother

\begin{abstract}
Modifications of the equations of ideal fluid dynamics with advected quantities are introduced that allow selective decay of either the energy $h$ or the Casimir quantities $C$ in the Lie-Poisson formulation. The dissipated quantity (energy or Casimir, respectively) is shown to decrease in time until the modified system reaches an equilibrium state consistent with ideal energy-Casimir equilibria, namely $\delta(h+C)=0$. The result holds for Lie-Poisson equations in general, independently of the Lie algebra and the choice of Casimir. This selective decay process is illustrated with a number of examples in 2D and 3D magnetohydrodynamics (MHD). 
\end{abstract}

\maketitle

\tableofcontents



\newpage

\section{Introduction}

Historically, the hypothesis of selective decay in MHD turbulence assumed that the total energy was to be minimized, subject to the conservation of certain ideal invariants. This hypothesis was consistent with the observed long term evolution of freely decaying MHD turbulence at high magnetic Reynolds number, $\mathrm{R}_\mathrm{m} = {U L}/{\eta}$, where $U$ is a typical velocity scale of the flow, $L$ is a typical length scale of the flow and $\eta$ is the magnetic resistivity. This situation was called an inverse cascade \cite{FPLM1975}, because the energy flux is predominantly toward small scales while the flux of the ideal invariant known as the magnetic helicity passes toward the larger scales and ostensibly creates the spiral structures observed in the decay of MHD turbulence. See \cite{MaMo1980,MoBa1999} for further historical discussion of the selective decay hypothesis for MHD turbulence. 


The aim of this paper is to { develop a geometric theory of} selective decay of either the energy or { the Casimirs (the invariants of the Lie-Poisson bracket of the ideal theory) and illustrate it for} the MHD Hamiltonian structure, following the work of  \cite{FGBHo2012} for geophysical fluid dynamics. We interpret the resulting modifications of the { ideal MHD} equations as a means of dynamically and nonlinearly parameterizing the interactions between disparate scales, by introducing new nonlinear pathways to dissipation, based on selective decay. 
Remarkably, the theory developed here for selective decay of \emph{either} the energy $h$ or the Casimir $C$ { contains} the standard energy-Casimir equilibria of the ideal equations, obtained from a critical point of the \emph{sum} $\delta(h+C)=0$ \cite{HMRW1985}. { Thus, our selective decay theory is always \emph{consistent} with the energy-Casimir equilibrium conditions, in that the energy-Casimir equilibria are also equilibria of the modified equations. However, the presence of the selective decay terms allows a new balance that enlarges the class of asymptotic states beyond those that satisfy the energy-Casimir equilibrium conditions associated with $\delta(h+C)=0$. In particular, the geometric selective decay process introduced here may in some cases tend towards states that satisfy only a \emph{subset} of the energy-Casimir equilibrium conditions. This is explained in the proof of the main result, Theorem \ref{SteadyStates} and is illustrated for the cases of compressible and incompressible MHD in Section \ref{MHD-sec}. }


\subsection*{Casimirs}

A Poisson manifold is a manifold $P$ with a Poisson bracket $\{\,\cdot\,,\,\cdot\,\}$ defined on the
space of smooth functions on $P$; see, e.g., \cite{MaRa1994}. A Poisson system with Hamiltonian $h : P \to \mathbb{R}$ yields the time-evolution of any smooth function $f : P \to \mathbb{R}$ by computing the solution curves of the dynamical equation $df/dt = \{f, h\}$.
Poisson systems often arise by Lie-group reduction of Hamiltonian systems with symmetry on Lie groups, in which case the Poisson bracket is called a Lie-Poisson bracket. Examples include the Euler equations of ideal incompressible fluid dynamics, for which the symmetry group is the particle-relabeling group \cite{ArKh1998}, and the equations for a heavy top, for which the symmetry group is the Euclidean group \cite{Ho2011}. 

\begin{definition}[Casimirs]\rm
Casimirs on a Poisson manifold $(P, \{\,\cdot\,,\,\cdot\,\})$ are functions $C$ that satisfy $\{C,h\}=0$, for all $h$, that is they are constant under the flow generated by the Poisson bracket for any choice of the Hamiltonian. The existence of Casimirs is thus due to the degeneracy of the Poisson bracket. In the case of reduction by symmetry on Lie groups, this degeneracy arises when passing from the canonical Hamiltonian formulation in terms of Lagrangian variables to the Lie--Poisson formulation in terms of symmetry-reduced variables.
\end{definition}

An example is the reduction of the Hamiltonian description of rigid body dynamics from the six-dimensional phase space $T^*SO(3)$ of the Euler angles for $SO(3)$ rotations, to the three-dimensional space of angular momenta in $\mathfrak{so}(3)^*\simeq\mathbb{R}^3$. For ideal fluids, the reduction is from the Lagrangian variables to the Eulerian variables, which are invariant under relabeling of Lagrangian particles. Thus, Casimir conservation is a property of the Lie--Poisson bracket that results from the reduction by symmetry, not the choice of Hamiltonian. Indeed, the Casimirs commute under the Lie--Poisson bracket with any Hamiltonian that is expressed in terms of the symmetry-reduced variables. This means the  motion generated by the Lie--Poisson bracket in the symmetry-reduced variables takes place along intersections of level sets of the Hamiltonian $h$ and a Casimir $C$. 

Casimirs have been used in stability analyses of fluid and plasma equilibria which extend traditional energy methods to the \textit{energy-Casimir method} \cite{Ar1969,HMRW1985}. This method applies in determining the stability of a certain class of equilibrium solutions $p_e\in P$ under the Poisson flow $df/dt = \{f, h\}$ defined on the Poisson manifold $P$ and generated by a Hamiltonian $h$. Namely, the energy-Casimir method supposes that there is a function $C$, the Casimir, which is constant under the flow generated by the Poisson bracket (since $\{C,h\}=0$ for all $h$) and that the equilibrium solution $p_e$ is a critical point of the sum $h_C:=h+C$, so that $\delta h_C=\langle Dh_C(p_e),\delta p\rangle =0$ for a nondegenerate pairing $\langle\,\cdot\,,\,\cdot\,\rangle$. 
Linear Lyapunov stability follows if the critical point $p_e$ is a local extremum of $h_C$, that is, if $\delta^2h_C(p_e)$ is positive definite or negative definite. This condition implies \textit{formal stability} since $\delta^2h_C(p_e)$ is conserved by the linearized equations around the equilibrium solution $p_e$ and if it is either positive definite or negative definite, then it defines a norm for Lyapunov stability, see \cite{HMRW1984,HMRW1985}. 
Nonlinear Lyapunov stability follows by a further argument, which is available when the functional $h_C$ is convex in the neighborhood of $p_e$. One may use the energy-Casimir method to seek stable equilibrium states. Such equilibrium states are called stable \emph{energy-Casimir} equilibria. Because of  the exchange symmetry of $h_C:=h+C$ under $C\leftrightarrow h$ these equilibrium states may be regarded as either extrema of the energy $h$ on a level set of a Casimir $C$, or vice versa, as extrema of a Casimir $C$ on a level set of the energy $h$.

\subsection*{Aim of the paper}

The aim of the present paper is to introduce a dissipative modification of the Lie--Poisson flow whose dynamics will tend toward conditions { that include} the energy-Casimir equilibria of the ideal unmodified equations, starting from any initial state on $P$.
In particular, the present paper investigates the effects of imposing selective decay of a { certain} Casimir while preserving the energy, and vice versa, imposing selective decay of energy while preserving the chosen Casimir. This is accomplished by using the Lie--Poisson structure of the ideal theory, and interpreting the resulting modifications of the equations as nonlinear pathways for selective dissipation that parameterize the observed effects of the interactions among disparate scales of motion. This type of modification is computable at a single time scale, so it may be useful in situations where it would be computationally prohibitive to rely on the slower, indirect effects of viscosity and other types of diffusivity which typically affect both the energy and the Casimir.  
In particular, the present paper takes the Casimir dissipation approach of \cite{FGBHo2012} further, by applying it { to 2D and 3D compressible and incompressible magnetohydrodynamics (MHD).} 

The new feature of the present paper is its introduction of modified Lie--Poisson equations that describe the selective decay of ideal fluids with \emph{advected quantities}.  Mathematically, ideal flows that advect fluid properties such as mass, heat and magnetic field may be described by the \emph{combined} actions of Lie groups on their dual Lie algebras and also on the vector spaces in which the advected quantities are defined \cite{HMR1998}. The Casimirs for such flows differ from the Casimirs of simple ideal fluid motion. These differences introduced by the advective flow of fluid properties yield new fluid equilibria and new nonlinear mechanisms for selective decay of either energy or Casimirs. In particular, the proof of Theorem \ref{SteadyStates} of the present paper shows that { under the geometric selective nonlinear decay process introduced here the flow tends toward conditions which include and extend the class of} energy-Casimir equilibria of the \emph{unmodified} fluid equations satisfying $\delta(h+C)=0$, for extrema of the energy $h$ on a level set of a Casimir $C$, or vice versa, for extrema of a Casimir $C$ on a level set of the energy $h$. These extrema may be either maxima or minima, depending on the choice of sign of a parameter $(\theta)$ and the sign of the Casimir appearing in the modified equations.

Various types of modifications of the Poisson bracket for Hamiltonian systems have been proposed in the literature  in order to include dissipation. This is usually accomplished by adding a symmetric bilinear form to the Poisson bracket, as initiated by \cite{Ka1984,Mo1984,Gr1984,Mo1986}. 
Specific classes of energy dissipation were introduced later in \cite{Br1991,Ka1991,BKMR1996,HoPuTr2008,BrElHo2008} by application of a double bracket.
See also  \cite{VaCaYo1989,Shep1990}, in which a modification of the transport velocity was used to impose energy dissipation with fixed Casimirs in an incompressible fluid. As the paper proceeds, we will comment further on the relationships of the present results with those of previous theories.

The theory we develop here uses the Lie-Poisson Hamiltonian framework for ideal fluids to treat either Casimir decay at fixed energy, or energy decay at fixed Casimirs. That is, the same theory is used here to treat selective decay of either quantity, so that the choice of which mechanism to investigate can be made, motivated for example by the effects seen in large-scale numerical simulations of the fully dissipative equations. The switch from decay of the chosen Casimir $C$ at fixed Hamiltonian $h$ to decay of the Hamiltonian $h$ at fixed Casimir $C$ is accomplished by a simple exchange  $C\leftrightarrow h$ in one of the key formulas, e.g., in equation \eqref{LP_form}. { In fact, either type of selective decay leads to the same equilibrium conditions.}

\paragraph{Selective decay of Casimirs.}

Selective decay of Casimirs is an effect that was first observed in numerical simulations of 2D incompressible turbulence cascades by \cite{MaMo1980} as the rapid decay of the enstrophy (turbulence intensity) while the energy stayed essentially constant. This observed disparity in the time scales for decay of the energy and the enstrophy in 2D turbulence has a profound effect on its energy spectrum.  In \cite{FGBHo2012},  selective decay by Casimir dissipation was introduced by modifying the vorticity equation, based on the well-known Lie--Poisson structure of the Hamiltonian formulation for vorticity dynamics in the case of 2D incompressible flows of ideal fluids   \cite{Ar1966,Ar1969,Ar1978,HMR1998}. In this framework, the earlier work of  \cite{VaCaYo1989,Shep1990} on selective decay of energy at fixed values of the Casimirs was recovered by the exchange of the Casimirs and the Hamiltonian in  formula \eqref{LP_form} for the modified vorticity dynamics. This earlier work studied selective decay for the purpose of finding stable equilibrium states. As we show here, imposing selective decay in fluid flows with advected quantities may also { enlarge the class of} stable energy-Casimir equilibrium states; see Theorem \ref{SteadyStates}. 

In 3D incompressible fluid turbulence, the energy tends to decay more rapidly that the Casimirs do. This is contrary to selective decay of turbulence in 2D, so a different modeling approach is required in 3D. Thus, a comprehensive theory must be capable of passing within the same framework from selective decay of the Casimir in 2D to selective decay of the energy in 3D. This is accomplished in the present theory by taking advantage of its exchange symmetry under $C\leftrightarrow h$.  

\subsection{Parameterizing subgridscale effects on macroscales}

In a previous paper \cite{FGBHo2012} the problem of parameterizing the interactions of disparate scales in fluid flows was addressed by considering a property of two-dimensional incompressible turbulence. The property  considered was a type of selective decay, in which a Casimir of the ideal formulation (enstrophy, in the case of 2D incompressible flows) was observed to decay rapidly in time, compared to the much slower decay of energy. That is, the Casimir was observed to decay, while the energy stayed essentially constant. 
The previous paper introduced a nonlinear fluid mechanism that produced the selective decay by enforcing Casimir dissipation at constant energy. This mechanism introduced an additional geometric feature into the description of the flow; namely, it introduced a Riemannian inner product on the space of Eulerian fluid variables. The resulting dissipation mechanism based on decay of enstrophy in 2D flows turned out to be related to the numerical method of \emph{anticipated vorticity} discussed in \cite{SaBa1981,SaBa1985}. 
Several examples were given and a general theory of selective decay was developed that used the Lie--Poisson structure of the ideal theory. A scale-selection operator allowed the resulting modifications of the fluid motion equations to be interpreted in these examples as parameterizing the nonlinear dynamical interactions between disparate scales. The type of modified fluid equations that was derived in the previous paper was also proposed for turbulent geophysical flows, where it is computationally prohibitive to rely on the slower, indirect effects of a realistic viscosity, such as in interactions between large-scale, coherent, oceanic flows and the much smaller eddies.  

The selective decay mechanism discussed in the previous paper was based on Casimir dissipation in the example of 2D incompressible flows, treated as a dynamical parameterization of the interactions between disparate scales. Following that example, the paper discussed the general theory of selective decay by Casimir dissipation in the Lie algebraic context that underlies the Lie--Poisson Hamiltonian formulation of ideal fluid dynamics, as explained in, e.g., \cite{HMR1998}. In particular, it developed the Kelvin circulation theorem and Lagrange-d'Alembert variational principle for Casimir dissipation. 
In the Lagrange-d'Alembert formulation, the modification of the motion equation to impose selective decay was seen as an energy-conserving constraint force.
Finally, the previous paper extended the Casimir dissipation theory to include fluids that possess advected quantities such as heat, mass, buoyancy, magnetic field, etc., by using the standard method of Lie--Poisson brackets for semidirect-product actions of Lie groups on vector spaces reviewed in \cite{HMRW1985}. The main subsequent examples were the rotating shallow water equations  and the 3D Boussinesq equations for rotating stratified incompressible fluid flows. 

\paragraph{Plan of the paper.}
The present paper pursues further the selective decay approach based on \cite{FGBHo2012}, whose main results are reviewed in the remainder of this Introduction. The formulations of selective decay of either Casimirs or energy on semidirect products is summarized in Section \ref{sec_SDP}. { The applications to compressible and incompressible 2D and 3D magnetohydrodynamics (MHD) in Section \ref{MHD-sec} illustrate the use of the method for fluid dynamics. In particular, we derive the modified MHD equations that enforce either selective decay of Casimirs at fixed energy, or vice versa.}

\subsection{Summary of key equations in \cite{FGBHo2012}}

Let us recall that ideal incompressible 2D fluid flows admit a Hamiltonian formulation in terms of a Lie--Poisson bracket $\{\,\cdot\,,\,\cdot\,\}_+$, given by \cite{Ar1966,Ar1969,Ar1978}
\begin{equation}
\frac{df( \omega )}{dt}
=\{f,h\}_+( \omega )
= \left\langle \omega , \left[ \frac{\delta f}{\delta \omega }, \frac{\delta h}{\delta \omega }\right] \right\rangle  
:= \int_\mathcal{D} \omega \left\{\frac{\delta f}{\delta \omega }, \frac{\delta h}{\delta \omega }\right\}dx\,dy \,.
\label{ArnoldLPB}
\end{equation}
Here $\omega$ is the vorticity of the flow, the bracket $\{\,\cdot\,,\,\cdot\,\}$ is the 2D Jacobian, written as 
$\{f,h\}=J(f,h)=f_xh_y-h_xf_y$, and the angle bracket $\langle\,\cdot\,,\,\cdot\,\rangle$ in (\ref{ArnoldLPB}) is the $L^2$ pairing in the domain $\mathcal{D}$ of the $(x,y)$ plane. For convenience, we shall take the domain $\mathcal{D}$ to be periodic, so we need not worry about boundary terms arising from integrations by parts.   
Two types of conservation laws are associated with the Hamiltonian formulation. The first one is the conservation of energy, i.e., the Hamiltonian $h( \omega )$. Conservation law of energy arises from the antisymmetry of the Lie--Poisson bracket as
\[
\frac{dh( \omega )}{dt}
=\{h,h\}_+ ( \omega )
= 0
\,,
\]
for any given choice of $h$. The second type of conservation law arises because the Lie--Poisson bracket has a kernel (i.e., is degenerate), which means there exist functions $C( \omega )$ for which 
\begin{equation}
\frac{dC( \omega )}{dt}
=\{C,h\}_+ ( \omega ) = 0
\,,
\label{Casimir-def}
\end{equation}
for any Hamiltonian $h(\omega)$.
Functions that satisfy this relation for any Hamiltonian are called \emph{Casimir functions}. (Lie called them distinguished functions, according to \cite{Ol2000}.) For example, the Casimirs for the Lie--Poisson bracket \eqref{ArnoldLPB} in the Hamiltonian formulation of 2D incompressible ideal fluid motion are
\[
C_\Phi (\omega) = \int_\mathcal{D} \Phi(\omega)\,dx\,dy 
\,,
\]
for any smooth function $\Phi$, \cite{Ar1966,Ar1969,Ar1978}.

Ideal 3D fluids also admit this type of Lie--Poisson bracket, given by
\[
\{f,g\}_+( \mathbf{u})=\int_ \mathcal{D} \mathbf{u} \cdot \left[  \frac{\delta f}{\delta \mathbf{u}  }, \frac{\delta g}{\delta \mathbf{u}  }\right] \,d^3x\, ,
\]
where $\mathbf{u}$, with $ \operatorname{div}\mathbf{u}  =0$, is the velocity and $[\,\cdot\,,\,\cdot\,]$ denotes the Lie bracket of vector fields, i.e., $[ \mathbf{u} , \mathbf{v} ]= \mathbf{v} \cdot \nabla \mathbf{u} - \mathbf{u} \cdot \nabla \mathbf{u} $.

\paragraph{Lie--Poisson brackets and Casimirs.} In this paper, we shall denote by $ \mathfrak{g}  $ a Lie algebra, with Lie brackets $ [\,\cdot\,,\,\cdot\,]$, and by $ \mathfrak{g}  ^\ast $ a space in weak nondegenerate duality with $ \mathfrak{g}$. That is, there exists a bilinear map (called a pairing) $ \left\langle\,\cdot\,, \,\cdot\,\right\rangle : \mathfrak{g}  ^\ast \times \mathfrak{g}  \rightarrow \mathbb{R}  $, such that for any $ \xi \in \mathfrak{g}  $, the condition $ \left\langle \mu , \xi \right\rangle = 0$, for all $\mu \in \mathfrak{g}  ^\ast $ implies $ \xi = 0 $ and, similarly, for any $ \mu \in \mathfrak{g}  ^\ast $, the condition $ \left\langle \mu , \xi \right\rangle = 0$ for all $\xi  \in \mathfrak{g}  $ implies $ \mu  = 0 $. Recall that $ \mathfrak{g}  ^\ast $ carries a natural Poisson structure, called the \emph{Lie--Poisson structure}, and given in terms of the pairing by
\begin{equation}\label{right_LP} 
\{f,h\}_+( \mu )= \left\langle \mu , \left[ \frac{\delta f}{\delta \mu }, \frac{\delta h}{\delta \mu }\right] \right\rangle,
\end{equation} 
(see, e.g., \cite{MaRa1994}). Here $f, g \in \mathcal{F} ( \mathfrak{g}  ^\ast )$ are real valued functions defined on $ \mathfrak{g}  ^\ast $, and $ {\delta f}/{\delta \mu }\in \mathfrak{g}  $ denotes the functional derivative of $f$, defined through the duality pairing $ \left\langle \,\cdot\,, \,\cdot\,\right\rangle $, by
\[
\left\langle \frac{\delta f}{\delta \mu }, \delta \mu \right\rangle = \left.\frac{d}{d\varepsilon}\right|_{\varepsilon=0} f( \mu + \varepsilon \delta \mu )\,.
\]
The Lie-Poisson bracket in \eqref{right_LP} is obtained by symmetry reduction of the canonical Poisson structure on the phase space $T^*G$ of the Lie group $G$ with Lie algebra $ \mathfrak{g}  $. The symmetry underlying this reduction is given by right translation by $G$ on $T^*G$. In the case of ideal fluid motion, this symmetry corresponds to relabeling symmetry of the Lagrangian in Hamilton's principle. 

\paragraph{Lie--Poisson (LP) equations.}
The Lie--Poisson (LP) equations with Hamiltonian $h: \mathfrak{g} ^\ast \rightarrow \mathbb{R}  $ are, by definition, the Hamilton equations  associated to the Poisson structure \eqref{right_LP}, i.e., 
\begin{equation}\label{right_LP-dyn} 
\frac{df}{dt} =\{f, h\}_+
\quad\hbox{for all}\quad
f \in \mathcal{F} ( \mathfrak{g}  ^\ast )
\,.
\end{equation} 
They are explicitly written as
\[
\partial _t \mu + \operatorname{ad}^*_{ \frac{\delta h}{\delta \mu }} \mu =0,
\]
where $ \operatorname{ad}^*_ \xi: \mathfrak{g}  ^\ast \rightarrow \mathfrak{g}  ^\ast $ is the coadjoint operator defined by $ \left\langle \operatorname{ad}^*_ \xi \mu , \eta \right\rangle = \left\langle \mu , [ \xi , \eta ] \right\rangle $.  One recalls that the coadjoint operator is equivalent to the Lie derivative, i.e., $\operatorname{ad}^*_ \xi \mu =\pounds_ \xi \mu$, when $\mu\in \mathfrak{g} ^\ast \simeq \Omega ^1\otimes {\rm dVol}$ is a 1-form density, as occurs in the case when $\mu$ is the momentum density in ideal fluid dynamics.

\paragraph{Casimir functions.} 
A function $C: \mathfrak{g}  ^\ast \rightarrow \mathbb{R} $ is called a \textit{Casimir function} for the Lie--Poisson structure (\ref{right_LP}) if it verifies $\{C, f\}_+=0$ for all functions $f \in \mathcal{F} ( \mathfrak{g}  ^\ast )$ or, equivalently 
\[
\operatorname{ad}^*_ { \frac{\delta C}{\delta \mu }}\mu=0
\,,\]
for all $ \mu \in \mathfrak{g}  ^\ast $. A Casimir function $C$ is therefore a conserved quantity for Lie--Poisson equations associated to any choice of the Hamiltonian $h$.

\paragraph{Symmetric bilinear form.} 
Below, we will denote by $ \gamma _ \mu $ a (possibly $ \mu $-dependent, $ \mu \in \mathfrak{g}  ^\ast $) \emph{symmetric bilinear form} $ \gamma _\mu : \mathfrak{g}  \times  \mathfrak{g}  \rightarrow \mathbb{R}$. This form is said to be \textit{positive} if
\[
\gamma _\mu ( \xi , \xi ) \geq 0, \quad \text{for all $ \xi \in \mathfrak{g}  $}\,. 
\]

\begin{definition}[Casimir-dissipative LP equation]\label{def_Casimir_Diss}  Given a Casimir function $C(\mu) $, for $\mu\in\mathfrak{g}^*$, a positive symmetric bilinear form $ \gamma_\mu$, and a real number $ \theta >0$, we consider the following modification of the Lie--Poisson (LP) dynamical equation \eqref{right_LP-dyn} to produce the Casimir \emph{dissipative} LP equation:
\begin{align}\label{LP_form} 
\frac{df(\mu) }{dt} 
&= \left\{ f,h \right\} _+
- \theta\, \gamma _\mu \left( \left[ \frac{\delta f}{\delta \mu } , \frac{\delta h}{\delta \mu } \right], \left[ \frac{\delta C}{\delta \mu } , \frac{\delta h}{\delta \mu } \right]\right),
\end{align} 
for arbitrary functions $ f,h: \mathfrak{g}  ^\ast \rightarrow \mathbb{R}$.

Equation \eqref{LP_form} yields the following equation for $\mu$,
\begin{equation}\label{Casimir_dissipation} 
\partial _t \mu + \operatorname{ad}^*_ { \frac{\delta h}{\delta \mu }  } \mu 
= \theta \operatorname{ad}^*_ { \frac{\delta h}{\delta \mu }  } 
 \left[ \frac{\delta C}{\delta \mu }, \frac{\delta h}{\delta \mu }\right] ^\flat  
\,,
\end{equation} 
where $ \flat : \mathfrak{g}  \rightarrow \mathfrak{g}  ^\ast $ is the flat operator associated to $ \gamma_\mu  $, that is, for $ \xi \in \mathfrak{g}  $, the linear form $ \xi ^\flat \in \mathfrak{g}  ^\ast $ is defined by $\left\langle  \xi ^\flat , \eta  \right\rangle = \gamma _\mu ( \xi , \eta )$, for all $\eta \in \mathfrak{g}  $.
\end{definition}

Note that the flat operator $ \flat $ need not be either injective or surjective. Note also that in equation \eqref{Casimir_dissipation} above, the flat operator is evaluated at $ \mu $. It is important to observe that the modification term depends on both the given Hamiltonian function $h$ and the chosen Casimir $C$.
It is convenient to write \eqref{Casimir_dissipation} as
\begin{equation}\label{Casimir_dissipation_tilde} 
\partial _t \mu + \operatorname{ad}^*_ { \frac{\delta h}{\delta \mu }  } \widetilde{\mu} =0,\quad\text{with modified momentum}\quad   \widetilde{ \mu }:= \mu+ \theta \left[ \frac{\delta h}{\delta \mu }, \frac{\delta C}{\delta \mu }\right] ^\flat.
\end{equation}

The energy is preserved by the dynamics of equations (\ref{LP_form}) and (\ref{Casimir_dissipation}), since we have
\[
\frac{dh(\mu) }{dt}= \left\{ h,h \right\} _+\,- \theta\, \gamma_\mu  \left( \left[ \frac{\delta h}{\delta \mu } , \frac{\delta h}{\delta \mu } \right], \left[ \frac{\delta C}{\delta \mu } , \frac{\delta h}{\delta \mu } \right]\right) =0.
\]
However, when $ \theta >0$ the Casimir function $C$ is dissipated since
\begin{equation}\label{Casimirdissipation} 
\frac{dC(\mu) }{dt}
= \left\{ C,h \right\} _+ -\, \theta\, \gamma _\mu \left( \left[ \frac{\delta C}{\delta \mu } , \frac{\delta h}{\delta \mu } \right], \left[ \frac{\delta C}{\delta \mu } , \frac{\delta h}{\delta \mu } \right]\right) 
=-\, \theta \left \| \left[ \frac{\delta C}{\delta \mu } , \frac{\delta h}{\delta \mu } \right]\right \| _\gamma ^2 ,
\end{equation} 
where $\| \xi \|^2_\gamma  :=\gamma_\mu  ( \xi , \xi )$ is the quadratic form (possibly degenerate) associated to the positive bilinear form $ \gamma_\mu  $.

{
\begin{remark}[Evolution of additional Casimirs]\rm 
If the Lie-Poisson bracket $\{\,\cdot\,,\,\cdot\,\} _+$ in a given case admits an additional Casimir $\widetilde{C}$, then  $\widetilde{C}$ will evolve according to \eqref{LP_form} as
\begin{equation}\label{Casimirdissipation2} 
\frac{d\widetilde{C}(\mu) }{dt}
= \left\{ \widetilde{C},h \right\} _+ 
-\, \theta\, \gamma _\mu \left( \left[ \frac{\delta \widetilde{C}}{\delta \mu } , 
\frac{\delta h}{\delta \mu } \right], \left[ \frac{\delta C}{\delta \mu } , \frac{\delta h}{\delta \mu } \right]\right)=-\, \theta\, \gamma _\mu \left( \left[ \frac{\delta \widetilde{C}}{\delta \mu } , 
\frac{\delta h}{\delta \mu } \right], \left[ \frac{\delta C}{\delta \mu } , \frac{\delta h}{\delta \mu } \right]\right),
\end{equation}
where we have used $ \left\{ \widetilde{C},h \right\} _+ =0$ because $\widetilde{C}$ is a Casimir.
\end{remark}
}

\begin{remark}[Left-invariant case]\rm Recall that the Lie--Poisson structure \eqref{right_LP} is associated to \textit{right} $G$-invariance on $T^*G$. We have made this choice because ideal fluids are naturally right-invariant systems in the Eulerian representation. Other systems, such as rigid bodies, are \textit{left} $G$-invariant. In this case, one obtains the Lie--Poisson brackets $\{f,g\}_-( \mu )= -\left\langle \mu , \left[ \frac{\delta f}{\delta \mu }, \frac{\delta h}{\delta \mu }\right] \right\rangle $ and this leads to the following change of sign in the Casimir-dissipative LP equation \eqref{Casimir_dissipation}: 
\begin{equation}\label{anticipated-ad_motion_left} 
\partial _t \mu - \operatorname{ad}^*_ { \frac{\delta h}{\delta \mu }  } \mu 
= \theta \operatorname{ad}^*_ { \frac{\delta h}{\delta \mu }  } 
 \left[ \frac{\delta C}{\delta \mu }, \frac{\delta h}{\delta \mu }\right] ^\flat  
\,.
\end{equation} 
The modified momentum is now
\begin{equation}\label{mod-mom-left} 
\widetilde{ \mu }:= \mu- \theta \left[ \frac{\delta h}{\delta \mu }, \frac{\delta C}{\delta \mu }\right] ^\flat
\end{equation}
and we have, in comparison with \eqref{LP_form},
\[
\frac{df(\mu) }{dt}=\left\{ f,h \right\} _- - \theta \gamma_\mu  \left( \left[ \frac{\delta f}{\delta \mu } , \frac{\delta h}{\delta \mu } \right], \left[ \frac{\delta C}{\delta \mu } , \frac{\delta h}{\delta \mu } \right]\right).
\]
\end{remark}

\begin{remark}[Relation with the metriplectic approach]{\rm By the following argument one may see that our selective decay model \eqref{LP_form} fits into the framework of the metriplectic dynamics initiated in the work of \cite{Mo1984}, \cite{Ka1984}, \cite{Gr1984}, \cite{Mo1986}.

Let $(P,\{\,,\})$ be a Poisson manifold, let $h\in C^\infty(P) $ be the Hamiltonian of the system and $C\in C^\infty(P) $ a Casimir function. The metriplectic dynamics is formulated as follows (see e.g., \cite{BlMoRa2012}):
\begin{equation}\label{metriplectic} 
\dot f=\{f,h\}+(f,C),\quad \text{for all $f \in C^\infty(P)$}, 
\end{equation} 
where the bracket $(f,g):= \left\langle \mathbf{d} f, \kappa (\mathbf{d} g) \right\rangle $ is $ \mathbb{R}  $-bilinear, symmetric, and positive (or negative) semidefinite, with $ \kappa :T^*P \rightarrow TP$ a vector bundle map. Moreover, it is assumed that $(f,h)=0$, for all $f\in C^\infty(P)$.
We shall now show that \eqref{LP_form} fits into the context of the metriplectic dynamics. In our case, the Poisson manifold is $P= \mathfrak{g}  ^\ast $ endowed with the Lie--Poisson bracket $\{\,\cdot\,,\cdot\,\}_\pm$. Upon replacing $C$ by an arbitrary function $g \in C^\infty( \mathfrak{g}  ^\ast )$,  the dissipative term in \eqref{LP_form} reads
\[
- \theta \gamma_\mu  \left( \left[ \frac{\delta f}{\delta \mu } , \frac{\delta h}{\delta \mu } \right], \left[ \frac{\delta g}{\delta \mu } , \frac{\delta h}{\delta \mu } \right]\right)= \theta \left\langle \frac{\delta f}{\delta \mu },\operatorname{ad}^*_{ \frac{\delta h}{\delta \mu }} \left( \left[ \frac{\delta g}{\delta \mu } , \frac{\delta h}{\delta \mu } \right] ^{\flat _ \mu }\right) \right\rangle = \left\langle \frac{\delta f}{\delta \mu }, \kappa_\mu  \left( \frac{\delta g}{\delta \mu } \right) \right\rangle = (f,g),
\]
where we have defined the vector bundle map $ \kappa _\mu :T ^\ast _\mu \mathfrak{g}  ^\ast \rightarrow T _\mu \mathfrak{g}  ^\ast $ by 
\begin{equation}\label{metriplectic2} 
\kappa _ \mu ( \xi ):= \theta \operatorname{ad}^*_{ \frac{\delta h}{\delta \mu }} \left( \left[ \xi , \frac{\delta h}{\delta \mu } \right] ^{\flat _ \mu }\right)
,
\end{equation} 
where $\flat _\mu: \mathfrak{g}\to \mathfrak{g}^*$ is the flat operator with respect to the pairing given by $\gamma_\mu$.
For this vector bundle map, the Lie-Poisson form  \eqref{LP_form} belongs to the class \eqref{metriplectic} of metriplectic systems.}
\end{remark}
\color{black}

\begin{definition}[Energy-dissipative LP equation]\label{def_energy_Diss} \rm
As discussed in \cite{FGBHo2012}, by simply exchanging $h$ and $C$ in the $ \theta $-term of equation \eqref{LP_form} one obtains an energy-dissipative LP equation that preserves the chosen Casimir $C$:
\begin{align}\label{LP_formE} 
\frac{df(\mu) }{dt} 
&= \left\{ f,h \right\} _+
- \theta\, \gamma _\mu \left( \left[ 
\frac{\delta f}{\delta \mu } , \frac{\delta C}{\delta \mu } 
\right], 
\left[ \frac{\delta h}{\delta \mu } , \frac{\delta C}{\delta \mu } \right]\right),
\end{align} 
for arbitrary functions $ f,h: \mathfrak{g}  ^\ast \rightarrow \mathbb{R}$.
\end{definition}

\begin{remark}[Energy-dissipative formulation]\rm
In the energy-dissipative formulation \eqref{LP_formE}, we have $dC/dt=0$ and energy decay given by
\begin{align}\label{Energy_dissipation} 
\frac{dh(\mu) }{dt} 
&= \left\{ h,h \right\} _+
- \theta\, \gamma _\mu \left( \left[ 
\frac{\delta h}{\delta \mu } , \frac{\delta C}{\delta \mu } 
\right], 
\left[ \frac{\delta h}{\delta \mu } , \frac{\delta C}{\delta \mu } \right]\right)
=
-\, \theta\, \left\| \left[ 
\frac{\delta h}{\delta \mu } , \frac{\delta C}{\delta \mu } 
\right]\right\|^2_\gamma
\,.
\end{align} 
By symmetry under the exchange $C\leftrightarrow h$, the two rates of decay are the same in \eqref{Casimirdissipation}  and \eqref{Energy_dissipation}.

An additional Casimir $\widetilde{C}$ would evolve according to equation \eqref{LP_formE} as
\begin{align}\label{addedCasimir-eqn} 
\frac{d\widetilde{C}(\mu) }{dt} 
&= 
- \theta\, \gamma _\mu \left( \left[ 
\frac{\delta \widetilde{C}}{\delta \mu } , \frac{\delta C}{\delta \mu } 
\right], 
\left[ \frac{\delta h}{\delta \mu } , \frac{\delta C}{\delta \mu } \right]\right).
\end{align} 
Equation \eqref{LP_formE} also leads to the following equation for $\mu$,
\begin{equation}\label{LP_form-hdis-mu} 
\partial _t \mu + \operatorname{ad}^*_ { \frac{\delta h}{\delta \mu }  } \mu 
= -\,\theta \operatorname{ad}^*_ { \frac{\delta C}{\delta \mu }  } 
 \left[ \frac{\delta C}{\delta \mu }, \frac{\delta h}{\delta \mu }\right] ^\flat  
\,,
\end{equation} 
where $ \flat : \mathfrak{g}  \rightarrow \mathfrak{g}  ^\ast $ is again the flat operator associated to $ \gamma_\mu  $. Of course, these equations follow from \eqref{Casimir_dissipation} by exchanging $C$ and $h$ in the $ \theta $-term.
\end{remark}

\paragraph{Energy-Casimir equilibria and their linear stability analysis}
\begin{theorem}[Energy-Casimir critical points]\rm$\,$\\
Energy-Casimir critical points $ \mu _e $ satisfying $\delta (h+C)( \mu _e )=0$ are steady states of the
modified LP equations \eqref{Casimir_dissipation} and \eqref{LP_form-hdis-mu} for
both the  Casimir-dissipative and the energy-dissipative cases. 
\end{theorem}

\begin{proof}
The energy-dissipative motion equation \eqref{LP_form-hdis-mu} is equivalent to
\begin{equation}\label{LP_form-hdis-mu-steady} 
\partial _t \mu + \operatorname{ad}^*_ { \frac{\delta (h+C)}{\delta \mu }  } \mu 
= -\,\theta \operatorname{ad}^*_ { \frac{\delta C}{\delta \mu }  } 
 \left[ \frac{\delta (h+C)}{\delta \mu }, \frac{\delta h}{\delta \mu }\right] ^\flat,
\end{equation} 
and the Casimir-dissipative motion equation \eqref{Casimir_dissipation} is equivalent to
\begin{equation}\label{LP_form-Cdis-mu-steady} 
\partial _t \mu + \operatorname{ad}^*_ { \frac{\delta (h+C)}{\delta \mu }  } \mu 
= -\,\theta \operatorname{ad}^*_ { \frac{\delta h}{\delta \mu }  } 
 \left[ \frac{\delta (h+C)}{\delta \mu }, \frac{\delta h}{\delta \mu }\right] ^\flat,
\end{equation}
Hence, the energy-Casimir stationarity condition $\delta (h+C)=0$ produces steady
states $\partial _t \mu_e=0$ of the modified LP equations in both cases.
\end{proof}

\begin{remark}[Linearized equations]\rm
The analysis of the linearized stability of energy-Casimir equilibria in the
energy-dissipative case, for example, proceeds from the linearization of equation
\eqref{LP_form-hdis-mu-steady} around $\mu_e$ for which $\delta (h+C)(\mu_e)=0$. 
Setting $\delta\mu=:\tilde{\mu}$ and defining
\[
\delta^2 (h+C)=\langle \tilde{\mu}, D^2(h+C)(\mu_e)\cdot\tilde{\mu}\rangle
=: \langle \tilde{\mu}\,,\,\tilde{\mu}^\sharp\rangle
\]
allows the linearization of \eqref{LP_form-hdis-mu-steady} to be written as
\begin{equation}\label{LP_form-hdis-mu-steady-stability} 
\partial _t \tilde{\mu} + \operatorname{ad}^*_ { \tilde{\mu}^\sharp  } \mu_e
= -\,\theta \operatorname{ad}^*_ { \frac{\delta C}{\delta \mu_e  }} 
 \left[ \tilde{\mu}^\sharp \,,\, \frac{\delta h}{\delta \mu_e }\right] ^\flat  
\,.
\end{equation} 
The linearized equation in the neighborhood of $\mu_e$ for the corresponding
Casimir-dissipative case is obtained from exchanging $h\leftrightarrow C$.

The presence of the $\theta$-term on the right hand side of \eqref{LP_form-hdis-mu-steady-stability}
alters the linearised spectrum of the energy-Casimir equilibria of the unmodified LP equations. 
An investigation of the effects of selective decay on the linear stability properties of the energy-Casimir 
equilibria is quite likely to be interesting. However, we shall defer this investigation to another work. 
 
\end{remark}


\subsection{Comparison with previous approaches}

\subsubsection{General energy dissipative systems}
Note that given a Poisson manifold $(P, \{\,,\})$, we can formulate the following energy dissipative system
\begin{equation}\label{general_formulation} 
\dot f=\{f,h\}- ((f,h)), \quad \text{for all $f \in C^\infty(P)$}, 
\end{equation}
where the bracket $((f,g))= \left\langle \mathbf{d} f, \Sigma ( \mathbf{d} g) \right\rangle$ is $ \mathbb{R}  $-bilinear, symmetric, and positive semidefinite, with $ \Sigma  :T^*P \rightarrow TP$ a vector bundle map. We can impose that a given function $C \in C^\infty(P)$ is preserved by the system by assuming $((f,C))=0$, for all $ f \in C^\infty(P)$. Our dissipative system \eqref{LP_formE} fits into this general picture by choosing
\[
((f,g))= \theta\, \gamma _\mu \left( \left[ 
\frac{\delta f}{\delta \mu } , \frac{\delta C}{\delta \mu } 
\right], 
\left[ \frac{\delta g}{\delta \mu } , \frac{\delta C}{\delta \mu } \right]\right),\quad \text{i.e.,} \quad  \Sigma _ \mu ( \xi )= \theta \operatorname{ad}^*_{ \frac{\delta C}{\delta \mu }} \left( \left[ \frac{\delta C}{\delta \mu }, \xi \right]^ { \flat _\mu }   \right),
\]
corresponding to the exchange $h\leftrightarrow C$ in equation \eqref{metriplectic2}.

\subsubsection{Double-bracket formulations}
There are apparent similarities in the Lie algebraic formulations of the present energy dissipation and the double-bracket formulations mentioned in the Introduction and reviewed, for example, in \cite{BlMoRa2012}. Indeed, for general Lie algebras the double-bracket dissipation equations can be written as
\begin{equation}\label{Lie_algebraic_DB} 
\frac{d f( \mu )}{dt}= \{f,h\}_+( \mu )-\theta  \gamma ^\ast \left( \operatorname{ad}^*_{ \frac{\delta k}{\delta \mu }} \mu ,  \operatorname{ad}^*_{ \frac{\delta f}{\delta \mu }} \mu  \right) ,
\end{equation} 
(compare with equation \eqref{LP_formE} to see the differences) where $ \gamma $ is a inner product on $ \mathfrak{g}  $, $ \gamma ^\ast $ is the inner product induced on $ \mathfrak{g}  ^\ast $, and $k: \mathfrak{g}  ^\ast \rightarrow \mathbb{R}  $ is a given function. One readily checks that Casimirs are preserved while, in the special case $k=h$, the energy dissipates. In that case, the equation of motion arising from \eqref{Lie_algebraic_DB} is given by
\begin{equation}\label{DB-motioneqn} 
\partial _t \mu + \operatorname{ad}^*_ { \frac{\delta h}{\delta \mu }  } \mu 
= \theta \operatorname{ad}^*_{\big(  \operatorname{ad}^\ast_{ \frac{\delta h}{\delta \mu }} \mu\big) ^\sharp}  \mu  
\,,
\end{equation} 
where $\sharp: \mathfrak{g}  ^\ast \rightarrow \mathfrak{g}  $ is the sharp operator associated to $ \gamma $. See \cite{BKMR1996,HoPuTr2008} for discussions of double-bracket  dissipation of energy. Note that \eqref{Lie_algebraic_DB} with $h=k$ fits into the framework of \eqref{general_formulation}.


\subsubsection{Comparison with double-bracket dissipation}\label{comp_DB} 
Formula \eqref{DB-motioneqn} for double-bracket dissipation coincides in some particular cases with equations \eqref{LP_form-hdis-mu}, obtained from our approach after exchanging the functions $C$ and $h$ in \eqref{LP_form}. More precisely, this coincidence occurs in the special case of \textit{quadratic Lie algebras}, i.e., Lie algebras that admit an ad-invariant inner product $ \gamma $, for example, semisimple Lie algebras. In this special case, taking the quadratic Casimir $ C( \mu )= \frac{1}{2} \gamma ( \mu , \mu )$, our equation \eqref{LP_formE} and the double bracket equation \eqref{DB-motioneqn} with $k=h$, coincide. In that case, $(\operatorname{ad}^*_{\xi}  \mu)^\sharp = -\,\operatorname{ad}_{\xi} \mu^\sharp$ and equation \eqref{DB-motioneqn} takes the double-bracket form,
\begin{equation}\label{DB-motioneqn-ss} 
\partial _t \mu^\sharp - \operatorname{ad}_ { \frac{\delta h}{\delta \mu }  } \mu^\sharp 
=  \theta \operatorname{ad}_{\big(  \operatorname{ad}_{ \frac{\delta h}{\delta \mu }} \mu^\sharp\big)}  \mu^\sharp  
=  \theta \left[\left[ \frac{\delta h}{\delta \mu }\,,\, \mu^\sharp\right] ,\, \mu^\sharp  \right]
\,.
\end{equation} 
This discussion generalizes to $ \mu $-dependent ad-invariant inner products $ \gamma _\mu $.

\paragraph{Examples.}
Examples of double-bracket formulations in fluid dynamics include \cite{VaCaYo1989,Shep1990}, in which a modification of the transport velocity was used to impose energy dissipation with fixed Casimirs. The Lie algebraic nature of this modification of the transport velocity becomes clear by rewriting the special case \eqref{DB-motioneqn} of the double bracket motion equation \eqref{Lie_algebraic_DB} as 
\begin{equation}\label{DB-motioneqn-rev} 
\partial _t \mu + \operatorname{ad}^*_{ v} \mu 
= 0 
\quad\hbox{with transport velocity}\quad
v =  \frac{\delta h}{\delta \mu }  
- 
 \theta {\big(  \operatorname{ad}^\ast_{ \frac{\delta h}{\delta \mu }} \mu\big) ^\sharp}
\,.
\end{equation}

\subsubsection{Comparison with the triple bracket formalism}Let us consider the following general form of a triple bracket on a manifold $P$
\[
\{f,g,h\}= \mathcal{C} ( \mathbf{d} f, \mathbf{d} g, \mathbf{d} h),\quad f,g,h \in C^\infty(P),
\]
where $ \mathcal{C}$ is an antisymmetric $3$-contravariant tensor field on $P$. Such a triple bracket can be used to formulate the following energy-dissipative system on the Poisson manifold $(P,\{\,,\})$
\begin{equation}\label{general_triple_bracket} 
\dot f=\{f,h\}-((f,h)), \quad \text{for all $f \in C^\infty(P)$} \quad \text{with} \quad ((f,g)):= \gamma \left(\mathcal{C} (\_\,, \mathbf{d} C, \mathbf{d} f) ,\mathcal{C} (\_\,, \mathbf{d} C, \mathbf{d} g) \right),
\end{equation}
Here, the expression $\mathcal{C} (\_\,, \mathbf{d} C, \mathbf{d} f)$ denotes the vector (i.e., linear form on $T^*P$) defined by
\[
\alpha \in T^*P\mapsto \mathcal{C} (\alpha , \mathbf{d} C, \mathbf{d} f)\in \mathbb{R} 
\]
For example, in local coordinates $\mathcal{C} (\_\,, \mathbf{d} C, \mathbf{d} f)$ is the vector
$\mathcal{C}^{ijk}\partial _j C \partial _k f$
and the energy-dissipative bracket becomes
\[
((f,g)):= \gamma \left(\mathcal{C} (\_\,, \mathbf{d} C, \mathbf{d} f) ,\mathcal{C} (\_\, \mathbf{d} C, \mathbf{d} g) \right)= \gamma _{il}\mathcal{C}^{ijk}\partial _j C \partial _k f\mathcal{C}^{lmn}\partial _m C \partial _n g
\,,
\]
where $ \gamma $ is a (possibly degenerate) metric.
In the particular  case when $P$ is a quadratic Lie algebra $\mathfrak{g}$ with ad-invariant inner product $ \gamma$, then $ \mathcal{C} ( \xi , \eta , \zeta ):= \gamma ( \xi , [ \eta , \zeta ])$ is antisymmetric. Identifying $ \mathfrak{g}  ^\ast $ with $ \mathfrak{g}  $ with the help of $ \gamma $, the triple bracket reads
\[
\{f,g,h\}= \gamma( \nabla f, [ \nabla g, \nabla h]),
\]
where $ \nabla f= ( \mathbf{d} f)^\sharp$ is the gradient of $f$ relative to $ \gamma $. This is the Lie algebraic generalization of the Nambu bracket (\cite{Nambu1973}) given in \cite{BBMo1991}. In this case, we have
\[
((f,h))= \gamma ([ \nabla C, \nabla f], [ \nabla C, \nabla h]),
\]
which coincides with the dissipation term in our system \eqref{LP_formE}, in the particular case of quadratic Lie algebras. This remark generalizes easily to a $\mu $-dependent ad-invariant inner product $ \gamma _\mu $. 
In general, however, our system \eqref{LP_formE} is not a special case of a triple bracket construction \eqref{general_triple_bracket}. 

Note that the triple bracket formulation can also be used to formulate a metriplectic dynamics
\begin{equation}\label{Casimir_triple_bracket}
\dot f=\{f,h\}+(f,C),\quad \text{for all $f \in C^\infty(P)$} 
\quad \text{with} \quad 
(f,g):=- \gamma \left(\mathcal{C} (\_\,, \mathbf{d} h, \mathbf{d} f) ,\mathcal{C} (\_\,, \mathbf{d} h, \mathbf{d} g) \right).
\end{equation} 
This construction was made in \cite[\S4.2]{BlMoRa2012} for the special case when $P$ is a quadratic Lie algebra $ \mathfrak{g}$ with ad-invariant inner product $ \gamma $ and with $ \mathcal{C} ( \xi , \eta , \zeta ):= \gamma ( \xi , [ \eta , \zeta ])$. In this particular case we have
\[
(f,g)= - \gamma ([ \nabla h, \nabla f], [ \nabla h, \nabla g])
\]
and \eqref{Casimir_triple_bracket} coincides with our equation \eqref{LP_form}.  However, this coincidence does not hold in general.

\color{black}
 
\paragraph{Outlook.}
In the remainder of this paper, we will first concentrate on Casimir dissipation at constant energy, and then treat the opposite case of energy dissipation at a constant Casimir by simply switching $h$ and $C$
 as in passing from equation \eqref{LP_form} to equation \eqref{LP_formE}. After this switch, we can reduce further to the double bracket form seen previously in the literature for the case of a quadratic Casimir and an $\operatorname{Ad}$-invariant inner product on the Lie algebra. To illustrate the method, the explicit formulas for selective energy decay at fixed values of the Casimir will be discussed in detail for MHD and compared with historical treatments of the selective decay hypothesis for MHD, such as \cite{BrCaPe1999}.

\subsubsection{Lagrange-d'Alembert variational principle} \label{LdA-princ}
Equations \eqref{Casimir_dissipation} and  \eqref{LP_form-hdis-mu} provide the
constraint forces that will guide the ideal MHD system into a particular class of
equilibria, by decreasing, respectively, either a particular choice of Casimir at
constant energy, or vice versa. The balance between the Casimir and energy that
occurs at a critical point of their sum determines the class of equilibria that is
achievable by a given choice of constraint force. The existence of a constraint
force that will dynamically guide an MHD system into a certain class of equilibria
(or preserve it once it has been obtained) may be useful in the design and control
of magnetic confinement devices. 

The Lagrange-d'Alembert variational principle extends Hamilton's principle to the
case of forced systems, including nonholonomically constrained systems
(\cite{Bl2004}). We now explain following \cite{FGBHo2012} how the
Casimir-dissipative LP equations \eqref{Casimir_dissipation} and energy-dissipative
LP equations \eqref{LP_form-hdis-mu} can be obtained from the Lagrange-d'Alembert
principle. Consider the Lagrangian $\ell: \mathfrak{g}  \rightarrow \mathbb{R}  $ related to $h$ via the Legendre transform, that is, we have
\[
h( \mu )= \left\langle \mu , \xi \right\rangle - \ell( \xi ), \quad  \mu := \frac{\delta \ell}{\delta \xi},
\]
where we have assumed that the second relation yields a bijective correspondence between $ \xi $ and $ \mu $.
In terms of  $\ell$, equation \eqref{Casimir_dissipation} for Casimir dissipation reads
\begin{equation}\label{EP_Casimir} 
\partial _t \mu + \operatorname{ad}^*_ { \xi  } \mu 
= \theta \operatorname{ad}^*_ { \xi  } 
 \left[ \frac{\delta C}{\delta \mu }, \xi \right] ^\flat  
\,,\quad \mu := \frac{\delta \ell}{\delta \xi}. 
\end{equation}
These equations can be obtained by applying the Lagrange-d'Alembert variational principle
\[
\delta\left[  \int_0 ^T \ell( \xi ) dt\right] + \theta \int_0 ^T \gamma \left( \left[ \frac{\delta C}{\delta \mu }, \xi \right] , [ \xi ,\zeta] \right) dt  =0, \quad\text{for variations} \quad\delta \xi = \partial _t \zeta -[ \xi , \zeta ],
\]
where $ \zeta \in \mathfrak{g}  $ is an arbitrary curve vanishing at $t=0,T$.
Thus, in the Lagrange-d'Alembert formulation, the modification of the motion equation to impose selective decay of the Casimir is seen as an energy-conserving constraint force.

\begin{remark}\rm
Similarly, the energy-dissipative LP equation \eqref{LP_form-hdis-mu} admits the variational formulation
\[
\delta\left[  \int_0 ^T \ell( \xi ) dt\right] + \theta \int_0 ^T \gamma \left( \left[ \xi , \frac{\delta C}{\delta \mu }\right] , \left [ \frac{\delta C}{\delta \mu },\zeta\right ] \right) dt  =0, \quad\text{for variations} \quad\delta \xi = \partial _t \zeta -[ \xi , \zeta ],
\]
where $ \zeta \in \mathfrak{g}  $ is an arbitrary curve vanishing at $t=0,T$.
\end{remark}

\subsubsection{Kelvin-Noether theorem} \label{Kelvin-thm} 
The well-known Kelvin circulation theorems for fluids can be seen as reformulations of Noether's theorem and, therefore, they have an abstract Lie algebraic formulation (the Kelvin-Noether theorems), see \cite{HMR1998}. We now discuss the abstract Kelvin circulation theorem for Casimir-dissipative LP equation \eqref{Casimir_dissipation}.   

In order to formulate the Kelvin-Noether theorem, one has to choose a manifold $ \mathcal{C} $ on which the group $G$ acts on the left and consider a $G$-equivariant map $ \mathcal{K} : \mathcal{C} \rightarrow \mathfrak{g}  ^{**}$, i.e. $\left\langle  \mathcal{K} (gc),\operatorname{Ad}_{g^{-1} } ^\ast   \nu  \right\rangle =\left\langle \mathcal{K} (c ), \nu  \right\rangle, \forall\; g \in G$. Here $gc$ denotes the action of $g \in G$ on $ c \in \mathcal{C}$ and $ \operatorname{Ad}^*_g$ denotes the coadjoint action defined by $\left\langle \operatorname{Ad}^*_g \mu , \xi \right\rangle = \left\langle \mu \operatorname{Ad}_g \xi \right\rangle $, where $ \mu \in \mathfrak{g}  ^\ast $, $ \xi \in \mathfrak{g}  $, and $\operatorname{Ad}_g$ is the adjoint action of $G$ on $ \mathfrak{g}  $. Given $c \in \mathcal{C}$ and $ \mu \in \mathfrak{g}  ^\ast $, we will refer to $ \left\langle \mathcal{K} (c), \mu \right\rangle $ as the \textit{Kelvin-Noether quantity} (\cite{HMR1998}). In application to fluids, $ \mathcal{C} $ is the space of loops in the fluid domain and $ \mathcal{K} $ is the circulation around this loop, namely
\[
\left\langle \mathcal{K} (c), \mathbf{u} \cdot d \mathbf{x} \right\rangle :=\oint_ c \mathbf{u} \cdot d\mathbf{x}.
\]

The Kelvin-Noether theorem for Casimir-dissipative LP equations is formulated as follows. 

\begin{proposition}\label{KN_casimir}  \rm
Fix $ c _0 \in \mathcal{C} $ and consider a solution $ \mu (t)$ of the Casimir-dissipative LP equation \eqref{Casimir_dissipation}. Let $g(t) \in G$ be the curve determined by the equation $ \frac{\delta h}{\delta \mu }= \dot g g ^{-1} $, $g(0)=e$. Then the time derivative of the Kelvin-Noether quantity $\left\langle \mathcal{K} ( g (t) c _0 ), \mu (t) \right\rangle$ associated to this solution is
\[
\frac{d}{dt} \left\langle \mathcal{K} ( g (t) c _0 ), \mu (t) \right\rangle =\theta \left\langle \mathcal{K} (g(t) c _0 ), \operatorname{ad}^*_ { \frac{\delta h}{\delta \mu }  } 
 \left[ \frac{\delta C}{\delta \mu }, \frac{\delta h}{\delta \mu }\right] ^\flat\right\rangle.
\]
\end{proposition} 

Note that $g(t)\in G$ is the motion in Lagrangian coordinates associated to the evolution of the momentum $ \mu (t)\in \mathfrak{g}  ^\ast $ in Eulerian coordinates. The $\theta$ term is an extra source of circulation with a double commutator. This term is absent in the ordinary Lie--Poisson case (i.e., for $ \theta =0$) and therefore in this case the Kelvin-Noether quantity $\left\langle \mathcal{K} ( g (t) c _0 ), \mu (t) \right\rangle$ is conserved along solutions.

\begin{corollary}\label{KN_erg}\rm
In the case of the energy-dissipative LP equation \eqref{LP_form-hdis-mu}, the Kelvin-Noether theorem is found from the exchange $C\leftrightarrow h$ to be
\[
\frac{d}{dt} \left\langle \mathcal{K} ( g (t) c _0 ), \mu (t) \right\rangle =\theta \left\langle \mathcal{K} (g(t) c _0 ), \operatorname{ad}^*_ { \frac{\delta C}{\delta \mu }  } 
 \left[ \frac{\delta h}{\delta \mu }, \frac{\delta C}{\delta \mu }\right] ^\flat\right\rangle.
\]
\end{corollary}

\subsection{Example: the rigid body}

The Lie--Poisson bracket on the dual Lie algebra of $\mathfrak{so}(3)$ may be written on $\mathbb{R}^3$ as 
\begin{equation}
\{F,h\}_-(\boldsymbol{\Pi} ) 
= -\,\boldsymbol{\Pi}\cdot \frac{\delta F}{\delta \boldsymbol{\Pi} }
\times \frac{\delta h}{\delta \boldsymbol{\Pi} }
\,,
\label{so3LPB}
\end{equation}
with $ \boldsymbol{\Pi}\in\mathbb{R}^3$.
The corresponding Lie--Poisson motion equation is, for the left-invariant case,
\[
\frac{d}{dt}\boldsymbol{\Pi} - \boldsymbol{\Pi} \times \frac{\delta h}{\delta \boldsymbol{\Pi} } =0\,.
\]
This equation describes the motion of a rigid body with Hamiltonian $h( \boldsymbol{\Pi}  )=\frac12\boldsymbol{\Pi}\cdot {\rm I}^{-1}\boldsymbol{\Pi}$ and symmetric positive-definite moment of inertia tensor ${\rm I}$ whose principle moments are assumed to be ordered as $I_1>I_2>I_3$.
The Casimir for this Lie--Poisson bracket is $C(\boldsymbol{\Pi} )=\frac{1}{2}  | \boldsymbol{\Pi} | ^2$ with $\frac{\delta C}{\delta \boldsymbol{\Pi} } = \boldsymbol{\Pi}$, and one may check that the Lie-Poisson bracket \eqref{so3LPB} yields $\{C,h\}=0$ for any Hamiltonian $h$.

\paragraph{Selective Casimir decay for the rigid body.}
The modified momentum induced by the principle of selective decay of Casimirs is found from equation \eqref{mod-mom-left} in this case to be
\[
\boldsymbol{\widetilde{\Pi}}
= \boldsymbol{\Pi} + \theta \left( \frac{\delta C}{\delta \boldsymbol{\Pi} } \times \frac{\delta h}{\delta \boldsymbol{\Pi} }  \right) ^\flat 
= \boldsymbol{\Pi} + \theta \left( \boldsymbol{\Pi} \times \frac{\delta h}{\delta \boldsymbol{\Pi} }  \right) ^\flat .
\]
The angular velocity of the rigid body is given by $ \frac{\delta h}{\delta \boldsymbol{\Pi} }= \boldsymbol{\Omega} ={\rm I}^{-1}\boldsymbol{\Pi}$. Choosing the usual inner product on $ \mathbb{R}  ^3 $ for the bilinear form $ \gamma_{\boldsymbol{\Pi}}  $ yields $ \flat =Id$, which implies from equation \eqref{anticipated-ad_motion_left} that
\begin{equation}\label{RB-CasDecayRate} 
\frac{d}{dt} \boldsymbol{\Pi}  - \boldsymbol{\Pi} \times \boldsymbol{\Omega} = \theta ( \boldsymbol{\Pi} \times \boldsymbol{\Omega} ) \times \boldsymbol{\Omega}
 , \quad \text{so that} \quad  
\frac{d}{dt} \frac{1}{2} \boldsymbol{\Pi} ^2 =- \theta | \boldsymbol{\Omega} \times \boldsymbol{\Pi} | ^2 \le0 \,.
\end{equation}
One might also have chosen the inner product $\gamma_{\rm I}$ associated with the inertia tensor ${\rm I}$, in which case
\[
\frac{d}{dt} \boldsymbol{\Pi}  - \boldsymbol{\Pi} \times \boldsymbol{\Omega} = \theta {\rm I}( \boldsymbol{\Pi} \times \boldsymbol{\Omega} ) \times \boldsymbol{\Omega} , \quad \frac{d}{dt} \frac{1}{2} \boldsymbol{\Pi} ^2 
=- \theta {\rm I} (\boldsymbol{\Omega} \times \boldsymbol{\Pi} ) \cdot (\boldsymbol{\Omega} \times \boldsymbol{\Pi})\le0\,.
\]

\paragraph{Selective energy decay  for the rigid body.}
Upon choosing instead to dissipate the energy at a fixed value of the Casimir and taking the usual inner product on $ \mathbb{R}  ^3 $ for the bilinear form $ \gamma_{\boldsymbol{\Pi}}  $ so that $ \flat =Id$,  the modified Lie--Poisson motion equation for selective decay of energy is found from equation \eqref{LP_form-hdis-mu}, written in the left-invariant case. Thus,  for the rigid body Hamiltonian $h( \boldsymbol{\Pi}  )=\frac12\boldsymbol{\Pi}\cdot {\rm I}^{-1}\boldsymbol{\Pi}$ and
Casimir $C(\boldsymbol{\Pi} )=\frac{1}{2}  | \boldsymbol{\Pi} | ^2 $, this becomes
\begin{equation}\label{RB-flow-erg-decay1} 
\frac{d}{dt} \boldsymbol{\Pi}  + \boldsymbol{\Omega} \times \boldsymbol{\Pi} 
= \theta\,  \boldsymbol{\Pi} \times ( \boldsymbol{\Pi} \times \boldsymbol{\Omega} ) 
\,,
\end{equation} 
which is the Landau-Lifshitz equation for spatially homogeneous dynamics of magnetization $( \boldsymbol{\Pi})$ at the microscopic scale \cite{LL1935}. Consequently, for this choice of the inner product given by the bilinear form $ \gamma_{\boldsymbol{\Pi}}  $ the rigid body energy decays as
\[
\frac{d}{dt} \left(\frac{1}{2} \boldsymbol{\Omega}\cdot \boldsymbol{\Pi}\right)
= \boldsymbol{\Omega}\cdot \frac{d}{dt}\boldsymbol{\Pi}
 =- \theta | \boldsymbol{\Omega} \times \boldsymbol{\Pi} | ^2 \le0 \,.
\]
\begin{remark}\rm
By the exchange symmetry of the dynamics of \eqref{LP_form} and \eqref{LP_formE} under $h\leftrightarrow C$, the energy of the rigid body decays at constant Casimir at the same rate as its Casimir decays at constant energy. The decay of either the energy or the Casimir ends at an equilibrium of the rigid body flow, at which the angular frequency $\boldsymbol{\Omega}$ and the $\boldsymbol{\Pi}$ angular momentum are aligned, so that $\boldsymbol{\Omega} \times \boldsymbol{\Pi}=0$, as expected from applying the energy-Casimir stability method in the example of the rigid body flow, \cite{HMRW1984}.
\end{remark}

\paragraph{Figures for selective decay of rigid-body energy at constant Casimir.}
Equation \eqref{RB-flow-erg-decay1} governs energy decay of the rigid body flow while preserving the Casimir whose level set defines angular momentum spheres in $ \mathbb{R}  ^3 $. 
The basins of attractions for the two North (green) and South (blue) least energy states are shown in Figures \ref{Fig5a-span_traj} for two different values of $\theta$. Along the basin boundaries in these figures, a slight change in the initial conditions may result in approaches to diametrically opposite equilibrium states asymptotically in time.

\begin{figure}[h!]
\begin{center}
\includegraphics[width=0.5\textwidth]{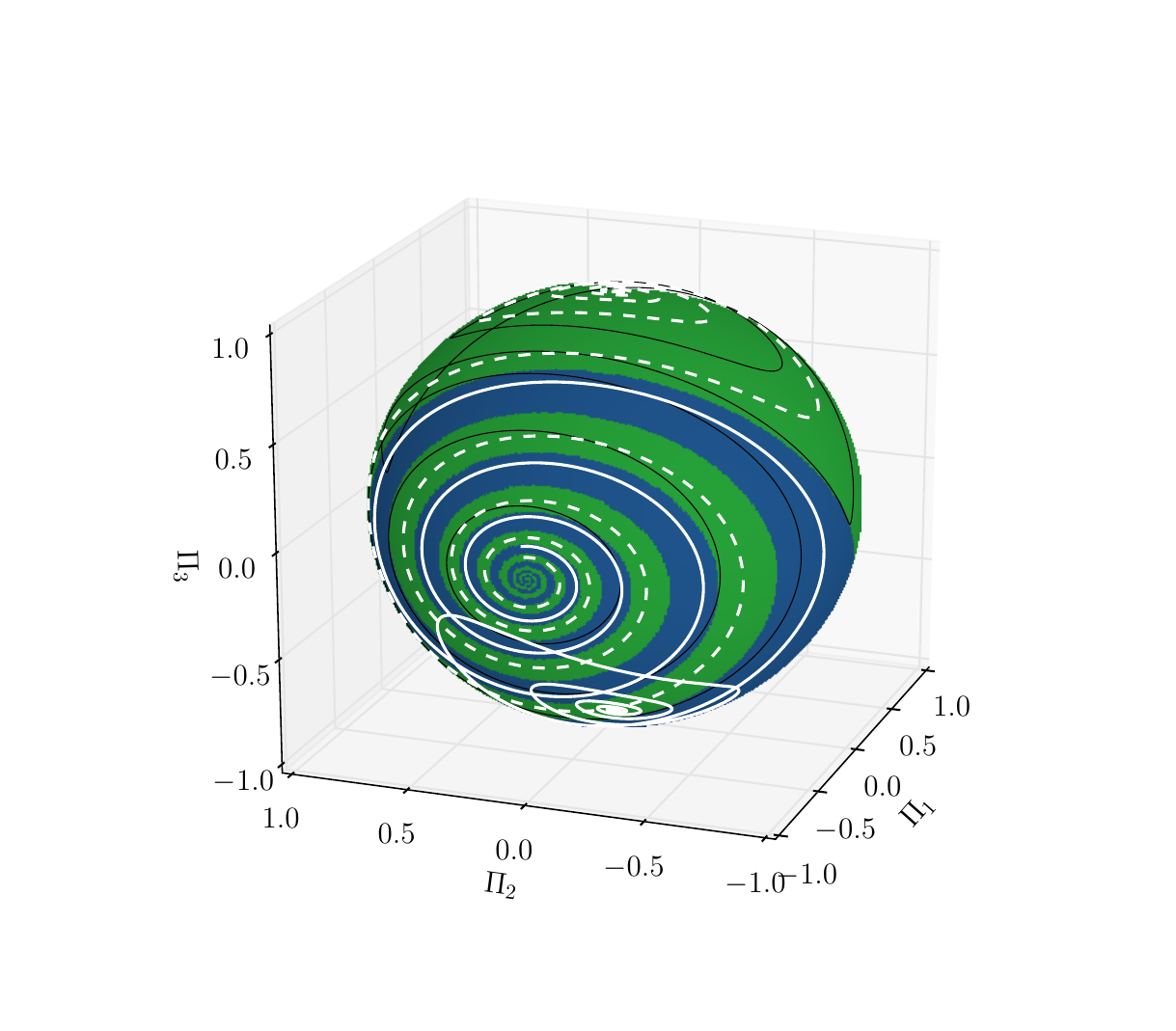}\!\!\!\!\!\!\includegraphics[width=0.5\textwidth]{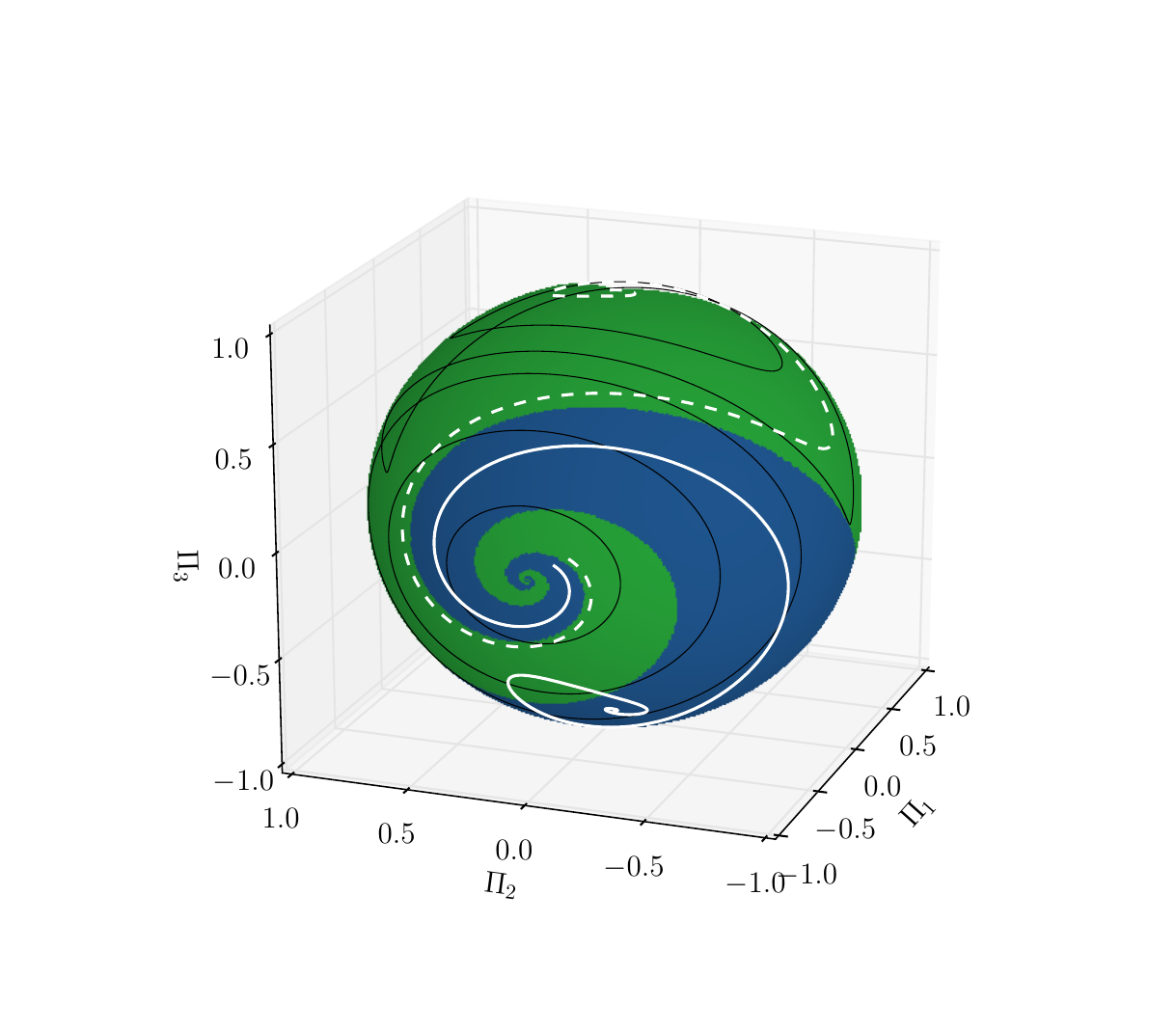}
\end{center}
\caption{\label{Fig5a-span_traj}\footnotesize 
\textit{Left}: For the solution curves of \eqref{RB-flow-erg-decay1} with $\theta=0.1$, this Figure shows the basins of attraction of the North (green) and South (blue) least energy states (of longest principle axis) lying at opposite points on the angular momentum sphere. Initial conditions starting in the blue (resp. red) region stay in the blue (resp. green) region. Along the basin boundaries, a slight change in the initial conditions may result in asymptotic approaches to diametrically opposite equilibrium states. \textit{Right}: For the solution curves of \eqref{RB-flow-erg-decay1} with $\theta=0.3$, this Figure shows the basins of attraction of the North (green) and South (blue) least energy states (of longest principle axis) lying at opposite points on the angular momentum sphere. Initial conditions starting in the blue (resp. green) region stay in the blue (resp. red) region. Along the basin boundaries, a slight change in the initial conditions may result in asymptotic approaches to diametrically opposite equilibrium states.}
\end{figure}




\section{Selective decay on semidirect products}\label{sec_SDP}
\subsection{Semidirect products}\label{subsec_SDP} 

The Hamiltonian structure of fluids that possess advected quantities such as heat, mass, buoyancy, magnetic field, etc., can be understood by using Lie--Poisson brackets for semidirect-product Lie algebras \cite{MaRaWe1984}.

In this setting, besides the Lie group configuration space $G$, one needs to include a vector space $V $ on which $G$ acts linearly. Its dual vector space $ V ^\ast $ contains the advected quantities. One then considers the semidirect product $G \,\circledS\, V$ with Lie algebra $ \mathfrak{g}  \,\circledS\, V$, and the Hamiltonian structure is given by the Lie--Poisson bracket \eqref{right_LP}, written on $ ( \mathfrak{g}  \,\circledS\, V) ^\ast $ instead of $ \mathfrak{g}  ^\ast $. We refer to \cite{MaRaWe1984}, \cite{HMR1998} for a detailed treatment. Given a Hamiltonian function $h=h( \mu ,a)$ with $h: ( \mathfrak{g}  \,\circledS\, V) ^\ast \to \mathbb{R}$ one thus obtains the Lie--Poisson equations
\begin{equation}\label{LP_SDP} 
\partial _t (\mu ,a)+ \operatorname{ad}^*_ { \left( \frac{\delta h}{\delta \mu } , \frac{\delta h}{\delta a} \right) }( \mu, a)=0,
\end{equation} 
for $\mu (t)\in \mathfrak{g}  ^\ast $ and $a (t)\in V ^\ast $. More explicitly, making use of the expression for the $ \operatorname{ad}^*$-operator in the semidirect product case, these equations read
\begin{equation}\label{SDP_LP} 
\partial _t \mu + \operatorname{ad}^*_{ \frac{\delta h}{\delta \mu }} \mu   +\frac{\delta h}{\delta a}\diamond a =0, \quad \partial _t a + a \frac{\delta h}{\delta \mu } =0,
\end{equation} 
where the operator $ \diamond : V \times V^\ast \rightarrow \mathfrak{g}  ^\ast $ is defined by 
\begin{equation}\label{diamond-def}
\left\langle v \diamond a, \xi \right\rangle := -\left\langle a \xi ,v\right\rangle
, \quad \text{for all $v \in V$, $a \in V ^\ast $, and $ \xi \in \mathfrak{g}  $}, 
\end{equation} 
and $ a \xi\in V ^\ast  $ denotes the (right) Lie algebra action of $ \xi \in \mathfrak{g}  $ on $ a \in V ^\ast $.

\paragraph{Casimir dissipation for semidirect products.} 
From the Lie algebraic point of view the direct generalization of \eqref{Casimir_dissipation} to semidirect product Lie groups would be
\begin{equation}\label{Casimir_dissipation_SDP} 
\partial _t (\mu ,a)+ \operatorname{ad}^*_ { \left( \frac{\delta h}{\delta \mu } , \frac{\delta h}{\delta a} \right) } (\mu ,a) 
= \theta \operatorname{ad}^*_ {\left(  \frac{\delta h}{\delta \mu } , \frac{\delta h}{\delta a} \right) }  \left[ \left( \frac{\delta C}{\delta \mu },\frac{\delta C}{\delta a }\right) , \left( \frac{\delta h}{\delta \mu },  \frac{\delta h}{\delta a } \right) \right] ^\flat,
\end{equation} 
where the flat operator $ \flat : \mathfrak{g}  \times V \rightarrow \mathfrak{g}  ^\ast \times V ^\ast $ is associated to a positive symmetric bilinear map $ \gamma_{( \mu , a)} :(\mathfrak{g}  \times V) \times ( \mathfrak{g}  \times V) \rightarrow \mathbb{R}$. Using the expression $[( \xi , v), ( \eta , w)]=([ \xi , \eta ], v \eta -w \xi )$
for the Lie bracket on $ \mathfrak{g}  \,\circledS\, V$, we can write \eqref{Casimir_dissipation_SDP} as
\begin{equation}
\label{Casimir_dissipation_SDP-system}
\partial _t (\mu ,a)+ \operatorname{ad}^*_ { \left( \frac{\delta h}{\delta \mu } , \frac{\delta h}{\delta a} \right) } (\widetilde{\mu},\widetilde a)=0,
\end{equation}
in which \emph{both} $ \mu $ and $ a$ are modified as
\begin{equation}\label{both_momenta} 
(\widetilde \mu , \widetilde a)= ( \mu , a)
+ 
\theta \left(\left[ \frac{\delta h}{\delta \mu }, \frac{\delta C}{\delta \mu }\right] , \frac{\delta h}{\delta a} \frac{\delta C}{\delta \mu }-\frac{\delta C}{\delta a} \frac{\delta h}{\delta \mu } \right) ^\flat
.
\end{equation} 
By using the formula $\operatorname{ad}^*_{( \xi , v)}( \mu , a)= (\operatorname{ad}^*_ \xi \mu + v \diamond a, a \xi)$ in equation \eqref{Casimir_dissipation_SDP-system}, one finds the explicit Casimir-dissipative system
\begin{equation}\label{Casimir_diss_SDP_explicit}
\partial _t \mu + \operatorname{ad}^*_{ \frac{\delta h}{\delta \mu }} \widetilde{\mu}   +\frac{\delta h}{\delta a}\diamond \widetilde{a}=0, \qquad \partial _t a+ \widetilde{a} \frac{\delta h}{\delta \mu }=0.
\end{equation}
When $ \gamma $ is diagonal on the Cartesian product $ \mathfrak{g}  \times V$, we can write $( \xi , v) ^\flat =( \xi ^\flat , v ^\flat )$ and \eqref{Casimir_diss_SDP_explicit} can be written explicitly as
\begin{equation}\label{explicit_SDP_Casimir}
\left\{ \begin{array}{l}
\displaystyle\vspace{0.2cm}\partial _t \mu + \operatorname{ad}^*_{ \frac{\delta h}{\delta \mu }} \mu  + \frac{\delta h}{\delta a}\diamond a
+ \theta \operatorname{ad}^*_{ \frac{\delta h}{\delta \mu }}
\left[ \frac{\delta h}{\delta \mu }, \frac{\delta C}{\delta \mu }\right] ^\flat
+ \theta\, \frac{\delta h}{\delta a}\diamond
\left(\frac{\delta h}{\delta a} \frac{\delta C}{\delta \mu }-\frac{\delta C}{\delta a} \frac{\delta h}{\delta \mu } \right) ^\flat=0\\
\displaystyle\partial _t a+ a \frac{\delta h}{\delta \mu } 
+ \theta 
\left(\frac{\delta h}{\delta a} \frac{\delta C}{\delta \mu }-\frac{\delta C}{\delta a} \frac{\delta h}{\delta \mu } \right) ^\flat
\frac{\delta h}{\delta \mu }=0
\,.
\end{array}\right.
\end{equation}
%
One may verify that the modified 
semidirect-product Lie--Poisson system \eqref{Casimir_dissipation_SDP} dissipates the Casimir $C$ while keeping energy conserved, under the modification of both $\mu$ and $a$. Namely, one computes the Lie--Poisson form
\begin{align}\label{Casimir_dissipation_SDP1} 
\begin{split}
\frac{df(\mu, a) }{dt} &=\left\{ f,h \right\} _+(\mu, a) 
-\theta \gamma  \left(\left[ \left( \frac{\delta f}{\delta \mu },\frac{\delta f}{\delta a }\right) , \left( \frac{\delta h}{\delta \mu },  \frac{\delta h}{\delta a } \right) \right],
\left[ \left( \frac{\delta C}{\delta \mu },\frac{\delta C}{\delta a }\right) , \left( \frac{\delta h}{\delta \mu },  \frac{\delta h}{\delta a } \right) \right] \right),
\end{split}
\end{align} 
which for $f=h$ shows that the energy is conserved while for $f=C$ shows that the Casimir dissipates.

\medskip

\begin{remark}[A simplification for $\frac{\delta C}{\delta \mu }=0$]\label{remark_simplification_Casimir} $\,$\rm
Note that the modification of $ \mu $ in the system \eqref{Casimir_diss_SDP_explicit} implies a modification of the $ \mu $-equation only. However, a modification of $a$ alone will yield a modification of both the $\mu$- and $a$-equations. For example, if $\frac{\delta C}{\delta \mu }=0$ then equation \eqref{both_momenta}  reduces to
\begin{equation}\label{Cmu0_momenta} 
\widetilde \mu = \mu, \qquad  \widetilde a=  a
-\,\theta \left( \frac{\delta C}{\delta a} \frac{\delta h}{\delta \mu } \right) ^\flat
,
\end{equation} 
and equation \eqref{explicit_SDP_Casimir} simplifies to 
\begin{equation}\label{Casimir_diss_SDP_explicit1}
\partial _t \mu + \operatorname{ad}^*_{ \frac{\delta h}{\delta \mu }} \mu  
+\frac{\delta h}{\delta a}\diamond a 
= \theta \frac{\delta h}{\delta a}\diamond\left( \frac{\delta C}{\delta a} \frac{\delta h}{\delta \mu } \right) ^\flat
, \quad \partial _t a + a \frac{\delta h}{\delta \mu }
= 
\theta \left( \frac{\delta C}{\delta a} \frac{\delta h}{\delta \mu } \right) ^\flat \frac{\delta h}{\delta \mu }.
\end{equation}
\end{remark}

\paragraph{Energy dissipation for semidirect products.} Exchanging the role of $h$ and $C$ in the $ \theta $-term of \eqref{Casimir_dissipation_SDP}, we get the energy-dissipative LP equation 
which preserves the Casimir $C$ for semidirect product Lie groups,  
\begin{equation}\label{energy_dissipation_SDP}
\partial _t (\mu ,a)+ \operatorname{ad}^*_ { \left( \frac{\delta h}{\delta \mu } , \frac{\delta h}{\delta a} \right) } (\mu ,a) =\theta \operatorname{ad}^*_ {\left(  \frac{\delta C}{\delta \mu } , \frac{\delta C}{\delta a} \right) }  
 \left(\left[ \frac{\delta h}{\delta \mu }, \frac{\delta C}{\delta \mu }\right] , \frac{\delta h}{\delta a} \frac{\delta C}{\delta \mu }-\frac{\delta C}{\delta a} \frac{\delta h}{\delta \mu } \right) ^\flat
\end{equation} 
In Lie--Poisson form this becomes 
\begin{equation}\label{energy_dissipation_SDP1} 
\frac{df(\mu, a) }{dt}=\left\{ f,h \right\} _+(\mu, a)
- \theta \,\gamma  \left(\left[ \left( \frac{\delta f}{\delta \mu },\frac{\delta f}{\delta a }\right) , \left( \frac{\delta C}{\delta \mu },  \frac{\delta C}{\delta a } \right) \right],
\left[ \left( \frac{\delta h}{\delta \mu },\frac{\delta h}{\delta a }\right) , \left( \frac{\delta C}{\delta \mu },  \frac{\delta C}{\delta a } \right) \right] \right),
\end{equation} 
which for $f=h$ shows that the energy dissipates as, 
\begin{equation}\label{energy_diss_SDP_gen}
\frac{dh(\mu, a) }{dt} =
-\,\theta \left\| \left[ \frac{\delta h}{\delta \mu }, \frac{\delta C}{\delta \mu }\right] \right\|_\gamma^2
-\,\theta \left\| 
\frac{\delta h}{\delta a} \frac{\delta C}{\delta \mu }-\frac{\delta C}{\delta a} \frac{\delta h}{\delta \mu }  \right\|_\gamma^2\,,
\end{equation} 
while for $f=C$ equation \eqref{energy_dissipation_SDP1} shows that the Casimir is conserved under the dynamics of \eqref{energy_dissipation_SDP1}.

After using the formula $ \operatorname{ad}^*_{( \xi , v)}( \mu , a)= (\operatorname{ad}^*_ \xi \mu + v \diamond a, a \xi)$ for the coadjoint operator of the semidirect product $ \mathfrak{g}  \,\circledS\, V$, and assuming that $ \gamma $ is diagonal on the Cartesian product $ \mathfrak{g}  \times V$, the system \eqref{energy_dissipation_SDP} is explicitly given by
\begin{equation}\label{energy_diss_SDP_explicit}
\left\{ \begin{array}{l}
\displaystyle\vspace{0.2cm}\partial _t \mu + \operatorname{ad}^*_{ \frac{\delta h}{\delta \mu }} \mu  + \frac{\delta h}{\delta a}\diamond a
+ \theta \operatorname{ad}^*_{ \frac{\delta C}{\delta \mu }}
\left[ \frac{\delta h}{\delta \mu }, \frac{\delta C}{\delta \mu }\right] ^\flat
+ \theta\, \frac{\delta C}{\delta a}\diamond
\left(\frac{\delta h}{\delta a} \frac{\delta C}{\delta \mu }-\frac{\delta C}{\delta a} \frac{\delta h}{\delta \mu } \right) ^\flat=0\\
\displaystyle\partial _t a+ a \frac{\delta h}{\delta \mu } 
+ \theta 
\left(\frac{\delta h}{\delta a} \frac{\delta C}{\delta \mu }-\frac{\delta C}{\delta a} \frac{\delta h}{\delta \mu } \right) ^\flat
\frac{\delta C}{\delta \mu }=0
\,.
\end{array}\right.
\end{equation}

\begin{remark}[Simplifications for $\frac{\delta C}{\delta \mu }=0$]\rm When $ \frac{\delta C}{\delta \mu }=0$, the energy-dissipative system \eqref{energy_diss_SDP_explicit} simplifies to 
\begin{equation}\label{energy_diss_SDP_explicit-noCmu}
\partial _t \mu + \operatorname{ad}^*_{ \frac{\delta h}{\delta \mu }} \mu  + \frac{\delta h}{\delta a}\diamond a
= \theta \,\frac{\delta C}{\delta a}\diamond
\left(\frac{\delta C}{\delta a} \frac{\delta h}{\delta \mu } \right) ^\flat
\,, \quad 
\partial _t a+ a \frac{\delta h}{\delta \mu } 
=0\,,
\end{equation}  
cf. equation \eqref{Casimir_diss_SDP_explicit1} for the corresponding simplification in the Casimir-dissipative case. Note that contrary to the Casimir dissipative case (Remark \ref{remark_simplification_Casimir}), the advection equation is left unchanged. The LP form of \eqref{energy_diss_SDP_explicit-noCmu} may be obtained by substitution, to find
\begin{equation}\label{erg-dis1}
\frac{df}{dt} = \left\langle \frac{\delta f}{\delta \mu } \,,\,\partial _t \mu \right\rangle
+ \left\langle \frac{\delta f}{\delta a } \,,\,\partial _t a \right\rangle=
\left\{ f,h \right\} _+(\mu, a)
- \theta\,\gamma \left(   
\frac{\delta C}{\delta a} \frac{\delta h}{\delta \mu }    
 \,,\, \frac{\delta C}{\delta a}\frac{\delta f}{\delta \mu }
\right).
\end{equation} 
Setting $f=h$ in the final equation of \eqref{erg-dis1} gives the energy dissipation equation
\begin{equation}\label{erg-dis2}
\frac{dh( \mu , a)}{dt}  = -\,\theta\,\gamma \left(   
\frac{\delta C}{\delta a} \frac{\delta h}{\delta \mu }    
 \,,\, \frac{\delta C}{\delta a}\frac{\delta h}{\delta \mu }
\right)
=
-\,\theta\left\|\,\frac{\delta C}{\delta a}\frac{\delta h}{\delta \mu }\,\right\|^2_\gamma
\end{equation}
which may also be obtained by setting $f=h$ and $\frac{\delta C}{\delta \mu}=0$ in the modified energy-dissipative LP equation \eqref{energy_dissipation_SDP1}, to find
\begin{equation}
\frac{dh( \mu , a)}{dt}  = -\, \theta \left \|\left(0, \frac{\delta C}{\delta a} \frac{\delta h}{\delta \mu } \right)\right \| ^2 _\gamma 
.
\end{equation}
Finally, setting $f=C$ and using $ \frac{\delta C}{\delta \mu }=0$ in the final equation of \eqref{erg-dis1} shows that the energy-dissipative system \eqref{energy_diss_SDP_explicit-noCmu} preserves the Casimir $C$. 
\end{remark}

\subsection{Convergence to steady states of the unmodified LP equations}

In the discussions below, we shall assume that the solutions of the modified (dissipative) equations possess long-time existence. That is, we shall work formally from the viewpoint of mathematical analysis, and ignore the possibility of blow up in finite time.

\begin{theorem}[Steady states]\label{SteadyStates}\rm
For either Casimir-dissipative or energy-dissipative LP equations for semidirect product Lie groups, under the modified dynamics \eqref{Casimir_dissipation_SDP} or \eqref{energy_dissipation_SDP}, the dissipated quantity (Casimir or energy, respectively), assumed to be positive,%
\footnote{As discussed in \cite{FGBHo2012}, one may assume $C \geq 0$, knowing that if $C\ne0$ is indefinite, one may replace it in these formulas by its square, $C\to C^2$, since the squares of Casimirs are still Casimirs.
}
decreases in time until the modified system reaches a { set of states} {\color{green} that} { include the} energy-Casimir equilibria associated to the Casimir $C$, namely $\delta(h+C)=0$, independently of the Lie algebra and the choice of Casimir. 
\end{theorem}

\begin{proof} Although a shorter proof of this theorem can be given, we choose to present it in three different cases, depending on how the advected variables are treated. This allows us to make several relevant comments in the proof. These cases are the following: \\
\textbf{(I)} the advected variables $a$ are absent; \\
\textbf{(II)} all of the variables $\mu,a$ are modified; and \\
\textbf{(III)} the advected variables $a$ are present but are left unmodified.

\noindent\textbf{(I)} The first class is the case in which the advected variables $a$ are absent, so that $h=h( \mu )$.
In this case, for \eqref{Casimir_dissipation}, resp. \eqref{LP_form-hdis-mu}, we have
\[
\frac{d}{dt} C( \mu )= -\, \theta\, \left\| \left[ 
\frac{\delta h}{\delta \mu } , \frac{\delta C}{\delta \mu } 
\right]\right\|^2_\gamma
\quad\text{resp.}\quad \frac{d}{dt} h( \mu )= -\, \theta\, \left\| \left[ 
\frac{\delta h}{\delta \mu } , \frac{\delta C}{\delta \mu } 
\right]\right\|^2_\gamma
\,.
\]
Thus, if $h, C \geq 0$ 
and $ \gamma $ is nondegenerate, both solutions converge to an asymptotic state with
\begin{equation}\label{asymptotic_state1} 
\left[ 
\frac{\delta h}{\delta \mu } , \frac{\delta C}{\delta \mu }  
\right]=0.
\end{equation} 
This condition holds for steady states $ \mu _e $ that satisfy the energy-Casimir equilibrium condition $\delta(h+C)/\delta \mu=0$ at $ \mu = \mu _e $, independently of the Lie algebra and the choice of Casimir. Note also that \eqref{asymptotic_state1} means that the $ \theta $-term in  the modified equation \eqref{erg-dis1} tends to zero.
\medskip

\textbf{(A)} In the special case when an ad-invariant pairing $ \kappa$ exists (e.g. if $ \mathfrak{g}$ is semisimple) then $ \operatorname{ad}^*_ \xi \mu =[ \mu , \xi ]$ and if we choose the Casimir $C( \mu )=  \frac{1}{2} \kappa ( \mu , \mu )$, then 
\[
\left[ 
\frac{\delta h}{\delta \mu } , \frac{\delta C}{\delta \mu } 
\right]= -\operatorname{ad}^*_{\frac{\delta h}{\delta \mu }} \mu,
\]
and in this case the solutions of both the Casimir dissipative and energy dissipative LP equations converge to a steady state, for any choice of the Hamiltonian $h$ and for all equilibria, not just for energy-Casimir equilibria. This is the case for the rigid body and for the 2D ideal fluid.
\medskip

\textbf{(B)} The above setting is not the only one in which this occurs. For example, for ideal incompressible 3D fluids with the helicity Casimir $C=\int  \mathbf{u} \cdot \operatorname{curl} \mathbf{u} \,d^3x$, the condition \eqref{asymptotic_state1} becomes
\[
\left[ \mathbf{u} , \operatorname{curl} \mathbf{u}  \right] =0 \quad\text{i.e.,}\quad  \operatorname{curl}( \mathbf{u} \times \operatorname{curl} \mathbf{u} )=0,\quad \text{i.e.,} \quad  \mathbf{u} \times \boldsymbol{\omega} = \nabla p.
\]
These equilibria are the steady Lamb flows, in which the level sets of pressure $p$ form symplectic manifolds \cite{ArKh1998}. In this situation, both the energy-dissipative and Casimir-dissipative LP equations converge to a steady state of the unmodified equations. The latter holds for the case that the Casimir is taken to be helicity-squared, cf. the footnote above.
\medskip

\noindent
\textbf{(II)} For a semidirect product LP system in which all variables are modified, as in \eqref{Casimir_dissipation_SDP} and \eqref{energy_dissipation_SDP}, and if $ \gamma $ is nondegenerate on $ \mathfrak{g}  \times V$, the equations converge as in \eqref{energy_diss_SDP_gen} to a state with both 
\begin{equation}\label{asymptotic_state} 
\left[ \frac{\delta h}{\delta \mu }, \frac{\delta C}{\delta \mu }\right] =0 \quad\text{and}\quad\frac{\delta h}{\delta a} \frac{\delta C}{\delta \mu }-\frac{\delta C}{\delta a} \frac{\delta h}{\delta \mu } =0\,.
\end{equation} 
These conditions mean that the $ \theta $-term in  the modified equation \eqref{erg-dis1} tends to zero.
Again this pair of conditions is satisfied for steady states $( \mu _e , a_e )$ that satisfy the energy-Casimir equilibrium condition $\delta(h+C)/ \delta (\mu , a  ) =0$, at $(\mu,a)=  (\mu _e , a _e)$, independently of the Lie algebra and the choice of Casimir, provided $ \frac{\delta C}{\delta \mu }\ne0$. 

When $ \frac{\delta C}{\delta \mu }=0$, condition \eqref{asymptotic_state} reduces to 
\begin{equation}\label{asymptotic_state-4Cmu=0}
0 = \frac{\delta C}{\delta a} \frac{\delta h}{\delta \mu }\,,
\end{equation}
which is an equilibrium state of the selective decay equation of energy (resp. Casimir) for $\frac{\delta C}{\delta \mu }=0$, as given in \eqref{erg-dis2}.

The requirement $\frac{\delta C}{\delta \mu }=0$ restricts the choice of Casimirs for either the modified LP equations or  the energy-Casimir equilibrium conditions of the unmodified equations. However, this case still retains some physically important cases, such as magnetic helicity for MHD, discussed among the examples in the later sections of the paper, particularly for Example 2 in \S\ref{3dhimhd} and in \S\ref{compMHD}.

\noindent
\textbf{(III)} For semidirect product LP equations with variables $(\mu,a)$ in which only the momentum equation is modified
and for which $ \gamma $ is nondegenerate on $ \mathfrak{g}  $, the solution converges to a solution with
\begin{equation}\label{asymptotic_state3} 
\left[ \frac{\delta h}{\delta \mu }, \frac{\delta C}{\delta \mu }\right] =0.
\end{equation} 
Once more, this condition holds for the class of equilibria for which the energy-Casimir method applies; namely,  the criticality condition $\delta(h+C)=0$.
\end{proof}

\begin{remark}\label{remark-thm-diagram}\rm
Here is a summary sketch diagraming the various lines of reasoning used in the proof of Theorem \ref{SteadyStates}. 

In all cases, we have the following diagram
\begin{displaymath}
\begin{xy}
\xymatrix{
\delta (h+C)=0 \ar@{<=>}[r]  &\frac{\delta h}{\delta \mu }+ \frac{\delta C}{\delta \mu } =0
\;\;\& \;\;
\frac{\delta h}{\delta a} + \frac{\delta C}{\delta a}=0  
\ar@{=>}[r] \ar@{=>}[d]&\eqref{asymptotic_state}  & &\\
\text{steady state}\ar@{<=>}[r] &\operatorname{ad}^*_{\left(  \frac{\delta h}{\delta \mu }, \frac{\delta h}{\delta a}\right) }( \mu , a)=0\,,  & &
}\end{xy}
\end{displaymath}
where implications are denoted by $A\Longrightarrow B$,  and equivalences are denoted by $A\Longleftrightarrow B$.

In the special case $ \frac{\delta C}{\delta \mu }=0$, the diagram becomes
\begin{equation}\label{C_a_only} 
\begin{xy}
\xymatrix{
\delta (h+C)=0 \ar@{<=>}[r]  &{ \frac{\delta h}{\delta \mu}=0}
\;\;\& \;\;
\frac{\delta h}{\delta a} + \frac{\delta C}{\delta a}=0  
\ar@{=>}[r] \ar@{=>}[d]&{\eqref{asymptotic_state-4Cmu=0} } & &\\
\text{steady state}\ar@{<=>}[r] &\operatorname{ad}^*_{\left({\frac{\delta h}{\delta \mu}}
, \frac{\delta h}{\delta a}\right) }( \mu , a)=0\,.  & &
}\end{xy}
\end{equation}
In the case $h=h( \mu )$ the diagram becomes
\begin{displaymath}
\begin{xy}
\xymatrix{
\delta (h+C)=0 \ar@{<=>}[r]  & \frac{\delta h}{\delta \mu }+ \frac{\delta C}{\delta \mu }=0  \ar@{=>}[r] \ar@{=>}[d]& 
\eqref{asymptotic_state1} & &\\
\text{steady state}\ar@{<=>}[r] &\operatorname{ad}^*_{\frac{\delta h}{\delta \mu } } \mu =0\,.  & &
}\end{xy}
\end{displaymath}
Note the directions of the implications. Namely, the critical point conditions $\delta (h+C)=0$ imply that the energy, or Casimir, decay rate vanishes, but not necessarily vice versa. This will become clear, later, when we discover that the number of critical point conditions obtained from $\delta (h+C)=0$ may in some cases exceeds the number of asymptotically vanishing decay rate terms obtained from the modified equations, as in Remark \ref{FF-compMHD}, for example.
\end{remark}
\color{black}

\subsection{Examples}

\subsubsection{Heavy top}

The Hamiltonian for a top spinning under the influence of gravity is the sum of its kinetic and potential energies, 
\begin{eqnarray}
\label{ht-ham}
h(\boldsymbol{\Pi},\boldsymbol{\Gamma})
=
\underbrace{\
\frac{1}{2}\boldsymbol{\Pi}
\cdot
{\rm I}^{-1}\boldsymbol{\Pi}\
}_{\hbox{kinetic}}
\
+\
\underbrace{\
mg \boldsymbol{\chi}\cdot \boldsymbol{\Gamma}\,
}_{\hbox{potential}}
\,,
\label{HT-Hamiltonian}
\end{eqnarray}
in which $\boldsymbol{\Pi}\in\mathbb{R}^3$ is the body angular momentum, $\boldsymbol{\chi}\in\mathbb{R}^3$ is the vector in the body from its point of support to its centre of mass, $mg$ is its weight, $\boldsymbol{\Gamma}(t)=O^{-1}(t)\mathbf{\hat{z}}\in\mathbb{R}^3$ is the vertical direction, as seen from the body and $\rm I$ is the moment of inertia of the top. 
The derivatives of this Hamiltonian are
\[
\frac{\delta  h}{\delta  \boldsymbol{\Pi}}
=
{\rm I}^{-1}\boldsymbol{\Pi}
 =: 
 \boldsymbol{\Omega}
\quad\hbox{and}\quad
\frac{\delta  h}{\delta  {\boldsymbol{\Gamma} }}
 = 
mg \boldsymbol{\chi}
\,.
\]

The heavy-top equations of motion for $\boldsymbol{\Pi} (t) $ and $\boldsymbol{\Gamma} (t) $ emerge from this Hamiltonian and the following Lie--Poisson bracket on the dual of the Euclidean Lie algebra $\mathfrak{se}(3)\simeq\mathfrak{so}(3) \,\circledS\, \mathbb{R}^3 \simeq \mathbb{R}^3 \,\circledS\, \mathbb{R}^3 $, 
\begin{equation}
\frac{d}{dt}(\boldsymbol{\Pi},\boldsymbol{\Gamma})
= {\rm ad}^*_{\left(\frac{\delta h}{\delta\boldsymbol{\Pi}},\frac{\delta h}{\delta\boldsymbol{\Gamma}}\right)}(\boldsymbol{\Pi},\boldsymbol{\Gamma})\,.
\end{equation}
In matrix form, these equations are
\begin{equation} 
\frac{d}{dt}
    \begin{bmatrix}
   \boldsymbol{\Pi} \\
   \boldsymbol{\Gamma}
    \end{bmatrix}
    =
    \begin{bmatrix}
    \boldsymbol{\Pi}\times
   &    
    \boldsymbol{\Gamma} \times
      \\
    \boldsymbol{\Gamma} \times
& 0
    \end{bmatrix}
    \begin{bmatrix}
   \delta  h/\delta \boldsymbol{\Pi} \\
   \delta  h/\delta \boldsymbol{\Gamma}
    \end{bmatrix}
=
    \begin{bmatrix}
   \boldsymbol{\Pi}\times\boldsymbol{\Omega} 
   + \boldsymbol{\Gamma} \times mg \boldsymbol{\chi}\\
   \boldsymbol{\Gamma}\times\boldsymbol{\Omega} 
    \end{bmatrix}.
\label{EP-se3eqns-uwv}
\end{equation}

\begin{remark}[Two Casimirs]\rm
This Lie-Poisson bracket admits two Casimirs. Namely, $C_0(\boldsymbol{\Pi} ,\boldsymbol{\Gamma} )=\frac12|\boldsymbol{\Gamma}|^2$ and $C_1(\boldsymbol{\Pi} ,\boldsymbol{\Gamma} )=\boldsymbol{\Gamma}\cdot \boldsymbol{\Pi}$. 
The quantities $C_0$ and $C_1$ are the body representations of, respectively, the squared magnitude of the spatial unit vertical vector and the vertical component of the spatial angular momentum, both of which are conserved. Conservation of $C_0(\boldsymbol{\Pi} ,\boldsymbol{\Gamma} )$ is merely a geometrical property and its use in selective decay only produces a trivial equilibrium state with no motion.
On one hand, it is interesting to note that the use of $C_0$ in this case does yield a steady state. This is an illustration of the diagram \eqref{C_a_only} in which \eqref{asymptotic_state-4Cmu=0} implies a steady state condition. On the other hand, this fact is not true in the general setting of  \S\ref{a_general_class} (that contains 2D MHD and heavy top); since the use of $C( \mu , a)= \frac{1}{2} \kappa (a,a)$ does not yield a steady state. However, conservation of $C_1(\boldsymbol \Pi,\boldsymbol\Gamma)$ has physical content and it leads to interesting equilibrium states for the case of heavy top dynamics. 
\end{remark}

\paragraph{Casimir dissipation for the heavy top.}
Applying the general equation \eqref{Casimir_dissipation_SDP} for Casimir dissipation in semidirect product dynamics to the heavy top example yields the following modified equations, which are reminiscent of the tippe top equations in \cite{Bou2008} and references therein,
\begin{equation} \label{Heavytop_C1_decay}
\begin{aligned}
\frac{d\boldsymbol\Pi}{dt} + \boldsymbol \Omega\times \boldsymbol \Pi + mg \boldsymbol\chi\times \boldsymbol \Gamma  &= -\theta \boldsymbol\Omega\times(\boldsymbol\Omega\times\boldsymbol\Gamma) -\theta \left( (mg)^2\boldsymbol \chi\times (\boldsymbol\chi\times \boldsymbol \Gamma) - mg\boldsymbol\chi\times (\boldsymbol\Pi\times \boldsymbol\Omega)\right)\\
\frac{d\boldsymbol\Gamma}{dt} + \boldsymbol\Omega\times \boldsymbol\Gamma   &=-\theta \left ( mg\boldsymbol\Omega\times (\boldsymbol\chi\times \boldsymbol\Gamma) - \boldsymbol \Omega\times(\boldsymbol\Pi\times \boldsymbol\Omega)\right).
\end{aligned}
\end{equation} 
The associated Casimir decay rate is
\begin{align}
\frac{d}{dt}(\boldsymbol{\Gamma} \cdot\boldsymbol{\Pi} )
&=
-\,\theta 
 \left\|  \boldsymbol\Omega \times \boldsymbol{\Gamma} \right\|^2
-\,\theta 
 \left\|   \boldsymbol\Omega\times \boldsymbol{\Pi}  + mg \boldsymbol\chi \times \boldsymbol{\Gamma} \right\|^2
. \label{C1norm_decay}
\end{align}
As the energy dissipation is very similar we will discuss both selective decay modifications together.

\paragraph{Energy dissipation for the heavy top}
The same computation using \eqref{energy_dissipation_SDP} for energy dissipation gives
\begin{align}
	\begin{split}
	\frac{d\boldsymbol\Pi}{dt} + \boldsymbol \Omega\times \boldsymbol \Pi + mg \boldsymbol\chi\times \boldsymbol \Gamma
	&= -\theta \boldsymbol\Gamma\times(\boldsymbol\Omega\times\boldsymbol\Gamma) -\theta \left ( mg\boldsymbol \Pi\times (\boldsymbol\chi\times \boldsymbol \Gamma) - \boldsymbol\Pi\times (\boldsymbol\Pi\times \boldsymbol\Omega)\right),\\
	\frac{d\boldsymbol\Gamma}{dt} + \boldsymbol\Omega\times \boldsymbol\Gamma   &=-\theta \left (mg\boldsymbol\Gamma\times (\boldsymbol\chi\times \boldsymbol\Gamma) - \boldsymbol \Gamma\times(\boldsymbol\Pi\times \boldsymbol\Omega)\right ),\\
	\end{split}
	\label{Heavytop_E_decay}
\end{align}
and the associated energy decay rate, cf. \eqref{energy_diss_SDP_gen}
\begin{align}
\frac{dh}{dt} 
&=
-\,\theta 
 \left\|  \boldsymbol\Omega \times \boldsymbol{\Gamma} \right\|^2
-\,\theta 
 \left\|   \boldsymbol\Omega\times \boldsymbol{\Pi}  + mg \boldsymbol\chi \times \boldsymbol{\Gamma} \right\|^2.
 \label{Enorm_decay}
\end{align}
As expected, the energy decay rate is the same as the Casimir decay rate \eqref{C1norm_decay}. 
Therefore, in both cases, as $t\to\infty$ the system will tend to a state which will satisfy the following two equilibrium conditions,
\begin{align}
 \boldsymbol\Omega \times \boldsymbol{\Gamma}  = 0\quad\hbox{and}\quad \boldsymbol\Omega\times \boldsymbol{\Pi}  + mg \boldsymbol\chi \times \boldsymbol{\Gamma}  =0\,.
\label{equi_cond1+2}
\end{align}
\begin{remark}\rm
This example is a good illustration of Remark \ref{remark-thm-diagram}, with the additional feature that here the implications are also equivalences, because $\mathfrak{se}(3)$ is of the special form $ \mathfrak{g}  \,\circledS\, \mathfrak{g}  $, where $ \mathfrak{g}  $ is a quadratic Lie algebra.
\end{remark}

\subsubsection{2D incompressible MHD}\label{remark_2d_ideal_MHD} 
For 2D incompressible MHD with $ \mathbf{B} $ in the $(x,y)$ plane, which will be discussed in the next section, exactly the same situation arises and the conclusion again depends on which Casimir is used. (In the MHD case, the heavy top Casimir $\boldsymbol{\Pi} \cdot \boldsymbol{\Gamma} $ corresponds to the planar MHD Casimir $\int \omega A\,dx\,dy $,  and $\frac{1}{2} | \boldsymbol{\Gamma} | ^2 $ corresponds to $\int A ^2 \,dx\,dy $.) This conclusion also applies for 2D incompressible MHD with $B$ perpendicular to the plane, also discussed in the next section.

\subsubsection{A general class}\label{a_general_class}
 
Both of the previous examples belong to the same general class. These are semidirect products of the type $ \mathfrak{g}  \,\circledS\,V $, where $V= \mathfrak{g}  $ is acted on by the adjoint action and $ \mathfrak{g}  $ admits an Ad-invariant pairing $ \kappa $. In this case, there are the two Casimirs
\[
C(\mu , a)=\kappa ( \mu , a) \quad\text{and}\quad C( \mu , a)= \frac{1}{2} \kappa (a,a),
\]
the unmodified equations \eqref{SDP_LP} read
\[
\partial _t \mu + \left[ \mu , \frac{\delta h}{\delta \mu }\right] + \left[ a, \frac{\delta h}{\delta a}\right] =0, \qquad \partial _t a+ \left[ a, \frac{\delta h}{\delta \mu }\right] =0,
\]
and condition \eqref{asymptotic_state} becomes 
\begin{equation}\label{asymptotic_state_quadratic} 
\left[ \frac{\delta h}{\delta \mu }, \frac{\delta C}{\delta \mu }\right] =0 \quad\text{and}\quad\left[ \frac{\delta h}{\delta a}, \frac{\delta C}{\delta \mu }\right] -\left[ \frac{\delta C}{\delta a} ,\frac{\delta h}{\delta \mu } \right] =0.
\end{equation}

Therefore, when the first Casimir is used, the condition \eqref{asymptotic_state} is equivalent to a steady state condition, so both the energy dissipative and Casimir dissipative (with squared Casimir) converge to a steady state. The corresponding critical point condition $\delta(h+C)=0$ reads $ \frac{\delta h}{\delta \mu }+a=0$, $ \frac{\delta h}{\delta a}+ \mu =0$. If the second Casimir is used, then \eqref{asymptotic_state} reads $[a, \frac{\delta h}{\delta \mu }]=0$ which implies that $a$ reaches a time independent state.
In this case, any steady state of the unmodified equations verifies the condition \eqref{asymptotic_state}. The converse is not true in general, but it is true for the heavy top.

The critical point condition $\delta(h+C)=0$ requires $\frac{\delta h}{\delta \mu }=0$, so the corresponding energy-Casimir equilibrium is trivial.

If only the momentum $\mu $ is modified, then only the first Casimir should be used since the second one yields no changes in the equation. In this case, the asymptotic solution verifies $[a, \frac{\delta h}{\delta \mu }]=0$, which implies that $a$ reaches a time independent state.

\color{black}

\section{Main Example: Selective decay for MHD} \label{MHD-sec}

In the barotropic (resp. incompressible) magnetohydrodynamics (MHD) approximation, plasma motion in three dimensions is governed by the following system of equations, see \cite{HMRW1985} and references therein:
\begin{equation}\label{MHD} 
\rho (\partial _t \mathbf{u} + \mathbf{u} \cdot \nabla \mathbf{u} )= - \nabla p+ \mathbf{J} \times \mathbf{B} , \quad \partial _t \mathbf{B}=- \operatorname{curl}\mathbf{E} , \quad \partial _t \rho + \operatorname{div}(\rho  \mathbf{u} )=0, \quad \operatorname{div} \mathbf{B} =0, 
\end{equation} 
where $p=p( \rho )$, (resp., $ \operatorname{div} \mathbf{u} =0$). 
Here  $\mathbf{B}$ denotes the magnetic field, $\mathbf{J} := \operatorname{curl} \mathbf{B} $ is the electric current density and $ \mathbf{E} := - \mathbf{u} \times \mathbf{B} $ expresses  the electric field in a frame moving with the fluid.

The pressure $p$ in the barotropic case is a given function of the mass density $ \rho $: $p=p (\rho) $. In contrast, $p$ is determined for the incompressible case by requiring that the condition $ \operatorname{div}\mathbf{u}= 0$ be preserved in time.

The barotropic MHD equations \eqref{MHD} can be augmented by including the specific entropy $ \eta $ verifying the advection equation $ \partial _t \eta + \mathbf{u} \cdot \nabla \eta =0$, and by considering pressure $p$ as resulting from the First Law of Thermodynamics,
\[
de = \rho^{-2}p\,d\rho + T d\eta
\]
for a given equation of state $e=e( \rho , \eta )$ for the internal energy per unit mass. The resulting \emph{isentropic} MHD equations are given in \eqref{MHD-eta}, and their properties under selective decay will be treated in Section \ref{compMHD}.

\subsection{Selective decay for three-dimensional MHD}

\subsubsection{3D homogeneous incompressible MHD}\label{3dhimhd} 

Consider the { Lie--Poisson bracket}  \cite{{MoGr1980},HoKu1983a,HoKu1983b,HMRW1985}
\[
\{f,g\}_+( \mathbf{m} , \mathbf{B} )
= \int_ \mathcal{D} \mathbf{m} \cdot \left[ \frac{\delta f}{\delta \mathbf{m} }, \frac{\delta g}{\delta \mathbf{m} }\right] \,d^3x   
+ \int_ \mathcal{D} \left( \operatorname{curl} \left( \mathbf{B} \times \frac{\delta f}{\delta \mathbf{m} } \right) \cdot \frac{\delta g}{\delta \mathbf{B} }- \operatorname{curl} \left( \mathbf{B} \times \frac{\delta g}{\delta \mathbf{m} } \right) \cdot \frac{\delta f}{\delta \mathbf{B} }\right) \,d^3x,
\]
with $ \operatorname{div} \mathbf{B} =0$. For the Hamiltonian
\[
h( \mathbf{m} , \mathbf{B} )= \int_ \mathcal{D} \left( \frac{1}{2} | \mathbf{m} | ^2 + \frac{1}{2} | \mathbf{B} | ^2 \right) \,d^3x 
\,,
\]
the Lie--Poisson equations recover \eqref{MHD} in the incompressible case, upon redefining the pressure. 

\paragraph{3D MHD Casimirs.} 
Incompressible 3D MHD has two Casimirs, the cross helicity and the magnetic helicity
\[
C_1( \mathbf{m} , \mathbf{B} )
=\int _ \mathcal{D} \mathbf{m} \cdot \mathbf{B} \,\,d^3x \quad\text{and}\quad  C_2( \mathbf{B} )=\frac{1}{2} \int _ \mathcal{D} \mathbf{B} \cdot {\rm curl}^{-1}\mathbf{B} \,\,d^3x \,.
\]
Note that $C_2$ is well-defined for $ \operatorname{div} \mathbf{B} =0$ and $H ^1 ( \mathcal{D} )= H ^2( \mathcal{D} )=0$.

\paragraph{Example 1: Cross helicity.} { 
For the cross helicity}, the modified momenta are
\[
\widetilde{ \mathbf{m}}= \mathbf{m}+ \theta \operatorname{curl}( \mathbf{m} \times \mathbf{B}) , \quad \widetilde {\mathbf{B}}= \mathbf{B} + \theta  \left( \operatorname{curl} \mathbf{B} \times \mathbf{B} -\operatorname{curl} \mathbf{m} \times \mathbf{m} {- \nabla \phi} \right).
\]
{
In this case, equations \eqref{Casimirdissipation} and \eqref{Casimirdissipation2} imply
\begin{align}\label{dissip_3DMHD-xhelicity} 
\begin{split}
\frac{d}{dt} C_1( \mathbf{m} , \mathbf{B} )
&= - \theta \int _ \mathcal{D} \left |\operatorname{curl} \left( \mathbf{m} \times  \mathbf{B}\right)  \right | ^2 d ^3 x
- \theta \int_ \mathcal{D} \left | \operatorname{curl} \mathbf{B} \times \mathbf{B} 
- \operatorname{curl} \mathbf{m} \times \mathbf{m} - \nabla \phi \right |^2  d ^3 x\,,
\\
\frac{d}{dt} C_2( \mathbf{m} , \mathbf{B} )
&=
- \theta \int_ \mathcal{D} (- \mathbf{B} \times \mathbf{m} - \nabla \phi _1 ) \cdot ( \operatorname{curl} \mathbf{B} \times \mathbf{B} - \operatorname{curl} \mathbf{m} \times \mathbf{m} - \nabla \phi _2 ) d ^3 x  
\,.
\end{split} 
\end{align}

When the squared Casimir $C_1^2$ is considered in the first line in \eqref{dissip_3DMHD-xhelicity}, the solutions converge to state with
\begin{equation}\label{state_dissip_3DMHD-xhelicity}
{\rm curl}\, (\mathbf{m} \times \mathbf{B}) = 0
\quad\hbox{and}\quad
\operatorname{curl}\,(\operatorname{curl} \mathbf{B} \times \mathbf{B} 
- \operatorname{curl} \mathbf{m} \times \mathbf{m}) =0
\,.
\end{equation} 
which is a an equilibrium of the unmodified equations.
The sub-case,
\[
\operatorname{curl} \mathbf{B} \times \mathbf{B}= 0
\quad\hbox{and}\quad
\mathbf{m} = 0
\,,
\]
comprises the ``force-free'' equilibria, introduced in \cite{Wo1958,Wo1959,Wo1960} and discussed in the context of toroidal $z$ pinch operation in \cite{Ta1974,Ta1986}. See \cite{GrKa1969} for a review of the earliest work in this field and \cite{BrCaPe1999} for discussions of later work. Usually, these equations are associated with minimization of the energy at constant helicity $C_2$. However, here they are associated with selective decay of the cross helicity $C_1$ at constant energy. }

\paragraph{Example 2: Magnetic helicity.} When the magnetic helicity is chosen, the modified momenta are
\[
\widetilde{ \mathbf{m} }= \mathbf{m} =:\mathbf{u} 
\quad\text{and}\quad 
\widetilde{ \mathbf{B} }= \mathbf{B} - \theta (\mathbf{E}+ { \nabla \phi }) 
\,,\quad\text{with}\quad 
\mathbf{E} := -\,\mathbf{u} \times \mathbf{B} 
.
\]
The dissipative LP equations \eqref{Casimir_dissipation_SDP}  associated to the magnetic helicity read
\[
\partial _t \mathbf{u} + \mathbf{u} \cdot \nabla \mathbf{u} = - \nabla p+ \mathbf{J} \times (\mathbf{B} -\theta (\mathbf{E}+{ \nabla \phi })  ), \quad \partial _t \mathbf{B}= -\operatorname{curl}(\mathbf{E} + \theta \mathbf{u} \times (\mathbf{E}+ { \nabla \phi })  ),
\]
where $ \operatorname{div}  \mathbf{u} =0$, $\operatorname{div} \mathbf{B} =0$, and we recall that $ \mathbf{J} = \operatorname{curl} \mathbf{B} $, $\mathbf{E} =-\, \mathbf{u} \times \mathbf{B} $, {and $ \nabla \phi $ is such that $ \mathbf{E} + \nabla \phi $ is divergence free}.
{
In this case, we get
\begin{align}\label{dissip_3DMHD-xhelicity2} 
\begin{split}
\frac{d}{dt} C_1( \mathbf{m} , \mathbf{B} )
& = \theta \int _ \mathcal{D}  (\mathbf{u}\times\operatorname{curl} \mathbf{u}  +\operatorname{curl} \mathbf{B} \times \mathbf{B}- \nabla \phi _1 )
\cdot( \mathbf{E}+ {\nabla \phi}_2 )  d ^3 x\,,
\\\,\frac{d}{dt} C_2( \mathbf{m} , \mathbf{B} )
&= 
  - \,\theta \int_ \mathcal{D} |\mathbf{u} \times \mathbf{B}{- \nabla \phi} | ^2 \,d^3x  
=-\, \theta \int_ \mathcal{D} |\mathbf{E}+ {\nabla \phi} | ^2 \,d^3x  \,,
\end{split} 
\end{align} 
so the system tends to a state verifying $\mathbf{E}+ {\nabla \phi} = 0$, so that $\partial_t\mathbf{B}=0$.
}

\begin{remark}\rm
Although the fluid momentum $\mathbf{m}$ is not modified { in this example}, the modification of $ \mathbf{B} $ produces changes in both the $ \mathbf{u} $- and $ \mathbf{B} $-equations, as explained in \S\ref{subsec_SDP}. 
Hence, selective decay by Casimir dissipation in this case produces loss of magnetic helicity. This means the introduction of Casimir dissipation causes a loss of linkages in the $\mathbf{B}$ field lines, due to reconnection.
\end{remark}

\paragraph{Example 3: Energy-dissipative case -- magnetic Lamb surfaces.} 
Equation \eqref{energy_diss_SDP_explicit} implies 
{ the following constant $C$, energy-dissipative, modified 3D MHD equations}
\begin{align*} 
&\partial _t \mathbf{m} + \operatorname{curl} \mathbf{m} \times \frac{\delta h}{\delta \mathbf{m} } + \mathbf{B} \times  \operatorname{curl} \frac{\delta h}{\delta \mathbf{B} } = -\theta \operatorname{curl} \left( \operatorname{curl} \left( \frac{\delta h}{\delta \mathbf{m} }\times \frac{\delta C}{\delta \mathbf{m} }\right) \right) ^\flat \times \frac{\delta C}{\delta \mathbf{m} }  \\
& \qquad \qquad \qquad - \theta   \left( \operatorname{curl} \frac{\delta h}{\delta \mathbf{B} }\times \frac{\delta C}{\delta \mathbf{m} }- \operatorname{curl} \frac{\delta C}{\delta \mathbf{B} }\times \frac{\delta h}{\delta \mathbf{m} }- \nabla \phi \right) ^\flat \times \operatorname{curl} \frac{\delta C}{\delta \mathbf{B} } - \nabla p, \\
&\partial _t \mathbf{B} + \operatorname{curl} \left( \mathbf{B} \times \frac{\delta h}{\delta \mathbf{m} }\right)  + \theta \operatorname{curl} \left( \left(\operatorname{curl}  \frac{\delta h}{\delta \mathbf{B} }\times \frac{\delta C}{\delta \mathbf{m} }- \operatorname{curl} \frac{\delta C}{\delta \mathbf{B} }\times \frac{\delta h}{\delta \mathbf{m} }- \nabla \phi \right) ^\flat \times \frac{\delta C}{\delta \mathbf{m} }\right)  =0\,,
\end{align*} 
where $ \nabla \phi $ is such that the term inside the parenthesis is divergence free.
For the MHD Hamiltonian, {in the energy-dissipative case when the cross helicity $C_1$ is held constant the corresponding $C_1$-modified 3D MHD} equations read, 
\begin{align*} 
&\partial _t \mathbf{u} + \mathbf{u} \cdot \nabla \mathbf{u} + \mathbf{B} \times \mathbf{J} = \theta (\operatorname{curl} \operatorname{curl} \mathbf{E}) \times \mathbf{B} - \theta \left( \operatorname{curl} \mathbf{B} \times \mathbf{B} - \operatorname{curl} \mathbf{u} \times \mathbf{u}- \nabla \phi  \right) \times \operatorname{curl} \mathbf{u} - \nabla p,\\
&\partial _t \mathbf{B} + \operatorname{curl}( \mathbf{B} \times \mathbf{u} )+ \theta \operatorname{curl} \left(  \left( \operatorname{curl} \mathbf{B} \times \mathbf{B} - \operatorname{curl} \mathbf{u} \times \mathbf{u}- \nabla \phi  \right) \times \mathbf{B} \right) =0
\,.
\end{align*} 
As expected, the energy dissipates at constant cross helicity {$C_1$} as
\[
\frac{d}{dt} \int_ \mathcal{D}\left(  \frac{1}{2} | \mathbf{u} | ^2 + \frac{1}{2} | \mathbf{B} | ^2\right)  d ^3 x= - \theta \int_ \mathcal{D} | \operatorname{curl}( \mathbf{u} \times \mathbf{B} )| ^2 d ^3 x- \theta \int_ \mathcal{D} | \operatorname{curl} \mathbf{B} \times \mathbf{B} - \operatorname{curl} \mathbf{u} \times \mathbf{u} - \nabla \phi | ^2 d ^3 x
\,. 
\]
Thus, the solution once again converges to a steady state of the unmodified 3D incompressible MHD equations. {By a remarkable coincidence, in this case the cross helicity $C_1$ and the magnetic helicity Casimir $C_2$ are both \textit{conserved}, as shown by a direct computation using the modified 3D MHD equations above.}

{When the magnetic helicity $C_2$ is chosen to remain constant while the energy $h$ dissipates, the modified  3D MHD} equations become
\begin{align*}
&\partial _t \mathbf{u} + \mathbf{u} \cdot \nabla \mathbf{u} + \mathbf{B} \times \mathbf{J} = \theta \left( \mathbf{B} \times \mathbf{u} + \nabla \phi \right) \times \mathbf{B}    - \nabla p\\
&\partial _t \mathbf{B} + \operatorname{curl}( \mathbf{B} \times \mathbf{u} )=0. 
\end{align*}
In this case, the advection equation for $ \mathbf{B} $ remains unmodified. The energy dissipates as
\[
\frac{d}{dt} \int_ \mathcal{D}\left(  \frac{1}{2} | \mathbf{u} | ^2 + \frac{1}{2} | \mathbf{B} | ^2\right)  d ^3 x 
=  - \,\theta \int_ \mathcal{D} |\mathbf{u} \times \mathbf{B}- \nabla \phi | ^2 \,d^3x  
=-\, \theta \int_ \mathcal{D} |\mathbf{E}+ \nabla \phi | ^2 \,d^3x  
\,,\]
{and again the magnetic helicity Casimir $C_2$ and the cross helicity $C_1$ are both \textit{conserved}.
Thus, the system again tends to a state verifying $\mathbf{E}+ {\nabla \phi} = 0$, so that $\partial_t\mathbf{B}=0$.}
In this state, the velocity vector $\mathbf{u}$ and the magnetic field vector $ \mathbf{B}$ are both tangent to each level set of the electrical potential $\phi$. These level sets may be called \emph{magnetic Lamb surfaces}, in analogy to the well-known Lamb surfaces for incompressible Euler fluid equilibria. 

\paragraph{The effects of rotation.} Rotation can be easily included in the Lie--Poisson formulation, by considering the Hamiltonian
\[
h( \mathbf{m} , \mathbf{B} )= \int_ \mathcal{D} \left( \frac{1}{2} | \mathbf{m} - \mathbf{R} | ^2 + \frac{1}{2} | \mathbf{B} | ^2 \right) \,d^3x,
\]
where $ \operatorname{curl} \mathbf{R} = 2 \boldsymbol{\Omega} $ is the Coriolis parameter (i.e., twice the angular
rotation frequency). The equations for Casimir dissipation can be derived by using $ \frac{\delta h}{\delta \mathbf{m} }= \mathbf{m} - \mathbf{R} = \mathbf{u} $ and $ \frac{\delta h}{\delta \mathbf{B} }= \mathbf{B} $. For example, for the magnetic helicity, we get, exactly as before,
\[
\widetilde{ \mathbf{m} }= \mathbf{m} \quad\text{and}\quad\widetilde{ \mathbf{B} }= \mathbf{B} - \theta (\mathbf{E}+ { \nabla \phi }).
\]
The modified equations now read 
\begin{align*}
&\partial _t \mathbf{u} + \mathbf{u} \cdot \nabla \mathbf{u} + 2 \boldsymbol{\Omega} \times \mathbf{u} = - \nabla p+ \mathbf{J} \times (\mathbf{B} - \theta (\mathbf{E} + \nabla \phi ))\\
&\partial _t \mathbf{B}= -\operatorname{curl}(\mathbf{E} + \theta \mathbf{u} \times (\mathbf{E} + \nabla \phi ))
\,.
\end{align*}

\subsubsection{3D compressible isentropic MHD}\label{compMHD}

Here we include the specific entropy $ \eta $ in the equations.
The Lie--Poisson bracket is
\begin{equation}\label{LP_3DMHD_comp}
\begin{aligned} 
\{f,g\}_+( \mathbf{m} , \rho , \eta ,\mathbf{B} )&= \int_ \mathcal{D} \mathbf{m} \cdot \left[ \frac{\delta f}{\delta \mathbf{m} }, \frac{\delta g}{\delta \mathbf{m} }\right] \,d^3x  +\int_ \mathcal{D}  \rho\left( \frac{\delta g}{\delta \mathbf{m} }\cdot \nabla \frac{\delta f}{\delta \rho }- \frac{\delta f}{\delta \mathbf{m} }\cdot \nabla \frac{\delta g}{\delta \rho }  \right)\,d^3x \\
& \quad + \int_ \mathcal{D} \eta   \left( \operatorname{div} \left( \frac{\delta f}{\delta \eta } \frac{\delta g}{\delta \mathbf{m} }\right)- \operatorname{div} \left( \frac{\delta g}{\delta \eta } \frac{\delta f}{\delta \mathbf{m} }\right)  \right) \,d^3x \\
& \quad + \int_ \mathcal{D} \left( \operatorname{curl} \left( \mathbf{B} \times \frac{\delta f}{\delta \mathbf{m} } \right) \cdot \frac{\delta g}{\delta \mathbf{B} }- \operatorname{curl} \left( \mathbf{B} \times \frac{\delta g}{\delta \mathbf{m} } \right) \cdot \frac{\delta f}{\delta \mathbf{B} }\right) \,d^3x,
\end{aligned} 
\end{equation} 
with $ \operatorname{div} \mathbf{B} =0$,.
The Lie--Poisson equations with Hamiltonian
\begin{equation}\label{MHD-Ham} 
h( \mathbf{m} , \rho ,\eta , \mathbf{B} )= \int_ \mathcal{D} \left( \frac{1}{2\rho }| \mathbf{m} | ^2 + \rho e ( \rho, \eta  )+ \frac{1}{2}| \mathbf{B} |^2 \right) \,d^3x ,
\end{equation} 
where $e ( \rho , \eta )$ is the 
specific internal energy, recover \eqref{MHD}
with an additional advection equation for the specific entropy variable $ \eta $. Namely,
\begin{align}\label{MHD-eta} 
\begin{split}
\rho (\partial _t \mathbf{u} + \mathbf{u} \cdot \nabla \mathbf{u} ) = - \nabla p+ \mathbf{J} \times \mathbf{B} , \quad \partial _t \mathbf{B}=- \operatorname{curl}\mathbf{E}
\,,\\
\partial _t \rho + \operatorname{div}(\rho  \mathbf{u} )=0\,, \quad
\partial _t \eta + \mathbf{u} \cdot\nabla \eta =0\,, \quad 
\operatorname{div} \mathbf{B} =0
\,.
\end{split}
\end{align} 
Here again $\mathbf{B}$ denotes the magnetic field, $\mathbf{J} := \operatorname{curl} \mathbf{B} $ is the electric current density, $ \mathbf{E} := - \mathbf{u} \times \mathbf{B} $ expresses  the electric field in a frame moving with the fluid, and now the pressure $p(\rho,\eta)$ is determined from the equation of state $e(\rho,\eta)$ by the First Law, as $p(\rho,\eta)=\rho ^2 {\partial e}/{\partial \rho }$.

\paragraph{Example 1: Potential magnetic intensity.} 
Three-dimensional compressible isentropic MHD possesses a scalar material invariant called \emph{potential magnetic intensity} (PMI), defined as
\[
q_B:= \mathbf{B}/ \rho  \cdot \nabla  \eta
\,,\]
in analogy to the potential vorticity for isentropic compressible fluid flow.
Material invariance of the scalar quantity $q_B$ may be derived by combining the three MHD advection laws for magnetic flux, $\mathbf{B}\cdot d\mathbf{S}$, specific entropy, $\eta$, and mass, $\rho d^3x$. Combining the first two of these three advected quantities into an advected volume-form yields a material-invariant 3-form in addition to the mass 3-form, so that
\[
(\partial_t+\pounds_\mathbf{u} )(\mathbf{B}\cdot d\mathbf{S} \wedge d\eta)=0
= (\partial_t+\pounds_\mathbf{u} )(\rho\, d^3x)
\,.
\]
Hence, the scalar quantity $q_B$ is advected, according to
\[
(\partial_t+\pounds_\mathbf{u} )q_B = \partial_t q_B + \mathbf{u}\cdot\nabla q_B = 0
\]
That is, the quantity $q_B$ is a scalar material invariant  (i.e., the scalar $q_B$ is conserved on particles). 
The material invariance of $q_B$ implies an additional class of Casimir functions for 3D compressible MHD, given by  functionals of the potential $\mathbf{B}$-field
\[
C ( \mathbf{m} , \rho , \eta , \mathbf{B} )
= \int_ \mathcal{D} \rho\, \Phi \left(q_B \right) \,d^3x 
= \int_ \mathcal{D} \rho \Phi \left( \mathbf{B}/ \rho  \cdot \nabla  \eta \right) \,d^3x 
\,,\]
for an arbitrary smooth function, $\Phi$. Note that, necessarily, $\widetilde{\mathbf{m}} = \mathbf{m} $ with these Casimirs. So we may again use \eqref{Casimir_dissipation_SDP} to introduce  Casimir dissipation and \eqref{energy_dissipation_SDP} to introduce energy dissipation. Since $ \delta C/ \delta \mathbf{m} =0$, the energy-dissipative approach keeps all the advection equations unchanged, see \eqref{energy_diss_SDP_explicit}. Although the explicit equations of motion may have complicated expressions, from \eqref{erg-dis2}, we easily obtain the energy dissipation in this case as
\[
\frac{d}{dt}h( \mathbf{m} , \rho , \eta , \mathbf{B} )
= - \,\theta \int_ \mathcal{D} \left| \mathbf{u} \cdot \nabla \frac{\delta C}{\delta \rho }\right | ^2 
+ \left | \operatorname{div} \left( \mathbf{u} \frac{\delta C}{\delta \eta }\right)   \right | ^2 
+ \left |\, \mathbb{P} \left( \operatorname{curl} \frac{\delta C}{\delta \mathbf{B} }\times \mathbf{u}\right)  \right | ^2 d ^3 x.
\]

\paragraph{Example 2: Magnetic helicity.} Another Casimir for compressible MHD (either with or without the specific entropy variable $ \eta $) is given as in the case of incompressible MHD  by the magnetic helicity
\begin{equation}\label{magnetic-helicity}
C( \mathbf{m} , \rho , \eta , \mathbf{B} )=\frac{1}{2} \int_ \mathcal{D} \mathbf{B} \cdot \operatorname{curl} ^{-1} \mathbf{B} \, \,d^3x  \,. 
\end{equation}
\begin{enumerate}
\item
In the {\bf Casimir-dissipative case}, the variational derivatives of the magnetic helicity imply the modified variables $\widetilde{ \mathbf{m} }=\mathbf{m}$, $\widetilde{ \rho }= \rho $, $\widetilde{ \eta }= \eta $, and
\[
\widetilde{ \mathbf{B} }
= \mathbf{B} + \theta\, \mathbb{P}\left( \frac{\delta h}{\delta\mathbf{m}}\times \operatorname{curl} \frac{\delta C}{\delta \mathbf{B} } - \frac{\delta C}{\delta\mathbf{m}}\times \operatorname{curl} \frac{\delta h}{\delta \mathbf{B} } \right) 
= \mathbf{B} - \theta \,\mathbb{P} \mathbf{E}
\,,
\]
when $ \gamma $ is chosen as the standard dot product, as in \eqref{magnetic-helicity}. 

As with the incompressible case, 
we find the Casimir dissipation equations
\[
\rho (\partial _t \mathbf{u} + \mathbf{u} \cdot \nabla \mathbf{u} )
= - \nabla p+ \mathbf{J} \times (\mathbf{B} - \theta\,\mathbb{P}\mathbf{E})
\,, \quad 
\partial _t \mathbf{B}
= -\operatorname{curl}(\mathbf{E} + \theta\, \mathbf{u} \times \mathbb{P}\mathbf{E}),
\]
which results in the dissipation of magnetic helicity,
\begin{equation}\label{helicity-diss}
\frac{d}{dt} \frac{1}{2} \int_ \mathcal{D} \mathbf{B} \cdot \operatorname{curl} ^{-1} \mathbf{B} \,\,d^3x 
=-\, \theta \int_ \mathcal{D} |\mathbb{P}\mathbf{E}| ^2 \,d^3x
\,.
\end{equation}

\item
In the {\bf energy-dissipative case}, since $ \delta C/ \delta \mathbf{m} =0$, only the momentum equation is modified. It reads
\[
\partial _t \frac{\mathbf{m} }{ \rho }
+\operatorname{curl}  \frac{\mathbf{m} }{\rho }  \times \frac{\delta h}{\delta \mathbf{m} }
+ \nabla  \left( \frac{\mathbf{m}  }{\rho }\!\cdot \!\frac{\delta h}{\delta \mathbf{m} }\right) 
= - \nabla \frac{\delta h}{\delta \rho } + \frac{1}{\rho } \frac{\delta h}{\delta \eta } \nabla \eta  
+  \frac{1}{\rho } \operatorname{curl} \frac{\delta h}{\delta \mathbf{B} } \times \mathbf{B}
+ \frac{\theta }{\rho }(\mathbb{P}( \mathbf{B} \times \mathbf{u}) ) \times \mathbf{B}.
\]
Correspondingly, the energy dissipates as
\begin{equation}\label{ECdiss-maghel}
\frac{d}{dt} h( \mathbf{m} , \rho , \eta , \mathbf{B} )
= -\, \theta \int_ \mathcal{D}  \left | \mathbb{P}(\mathbf{B} \times \mathbf{u}) \right | ^2 d ^3 x
=- \,\theta \int_ \mathcal{D}  \left | \mathbb{P}\mathbf{E}  \right | ^2 d ^3 x.
\end{equation}
Thus, the asymptotic solution verifies $ \mathbb{P}\mathbf{E} = \mathbf{E}+ \nabla \phi  = 0$
with ${\rm div}(\mathbb{P}\mathbf{E})=0$.
\end{enumerate}

\begin{remark}[Force-free compressible MHD equilibria]\label{FF-compMHD}\rm 
When the Casimir is chosen to be the magnetic helicity in \eqref{magnetic-helicity}, the associated energy-Casimir equilibrium conditions arising from $ \delta (h+C)=0$ read
\[
\frac{\delta (h+C)}{\delta \mathbf{m} } =0 \quad\text{and}\quad \frac{\delta (h+C)}{\delta\mathbf{B} }=0
\,.
\]
These equilibrium conditions yield both $ \mathbf{u} =0$ and $ \mathbf{B} + \operatorname{curl} ^{-1} \mathbf{B} =0$ with ${\rm div}\mathbf{B}=0$. Thus, in this case, $ \delta (h+C)=0$ implies both the equations for the ``force-free'' equilibria,
\begin{equation}\label{ECcond-maghel}
\mathbf{J} \times \mathbf{B} = \mathbf{0}\
\;\;(\hbox{upon using }\mathbf{J} := \operatorname{curl}\mathbf{B})
\quad\hbox{and}\quad
\mathbb{P}\mathbf{E} =\mathbf{0}
\,.
\end{equation}

In contrast to the two energy-Casimir conditions in \eqref{ECcond-maghel}, the energy-dissipative approach for constant magnetic helicity implies in \eqref{ECdiss-maghel} only the single condition $\mathbb{P}\mathbf{E} =\mathbf{0}$ at equilibrium. Thus, as stated in Theorem \ref{SteadyStates}, the energy-dissipative equilibrium conditions are \emph{consistent with} the associated energy-Casimir equilibrium conditions. However, the latter may contain more conditions and thus be more restrictive. { Hence, the set of energy-Casimir equilibria is a subset of the corresponding energy-dissipative equilibria.}

For energy dissipation at constant magnetic helicity, the diagram with the various implications discussed in Remark \ref{remark-thm-diagram} after Theorem \ref{SteadyStates} reads
\begin{displaymath}
\begin{xy}
\xymatrix{
\delta (h+C)=0 \ar@{<=>}[r]  &\mathbf{u} =0\;\;\& \;\;\mathbf{B} + \operatorname{curl} ^{-1} \mathbf{B} =0   \ar@{=>}[r] \ar@{=>}[d]&\eqref{asymptotic_state-4Cmu=0}: \mathbb{P}\mathbf{E} =0\,, & &\\
\text{steady state}\ar@{<=>}[r] &\operatorname{ad}^*_{\left(  \frac{\delta h}{\delta \mathbf{m}  }, \frac{\delta h}{\delta\mathbf{B} }\right) }( \mathbf{m}  , \mathbf{B} )=0\,.  & &
}\end{xy}
\end{displaymath}
That is, the implications need not always be equivalences.
\end{remark}

\paragraph{Case 3: Cross helicity.} When the specific entropy variable $ \eta $ is absent, the \emph{cross-helicity} given by
\[
C( \mathbf{m} , \rho  , \mathbf{B} )=\int_ \mathcal{D} \frac{1}{\rho }\mathbf{m} \cdot \mathbf{B}\,\,d^3x\,
\]
is a Casimir. In this case, the modified variables are found from \eqref{both_momenta} to be 
\begin{align*} 
\widetilde{ \mathbf{m} }&= \mathbf{m} + \theta \left[ \frac{\delta h}{\delta \mathbf{m} }, \frac{\delta C}{\delta \mathbf{m} }\right] ^\flat 
= \mathbf{m} + \theta \Big[ \mathbf{u} , \rho ^{-1} \mathbf{B} \Big]^\flat 
,\\
\widetilde{ \rho }&=\rho + \theta \left( \frac{\delta C}{\delta \mathbf{m} }\cdot \nabla \frac{\delta h}{\delta \rho }- \frac{\delta h}{\delta \mathbf{m} }\cdot \nabla \frac{\delta C}{\delta \rho }  \right)^\flat  = \rho + \theta \left(\rho ^{-1} \mathbf{B} \cdot \nabla \left(  e+ \rho \frac{\partial e}{\partial \rho }- \frac{1}{2} | \mathbf{u} | ^2 \right)  - \mathbf{u}\cdot \nabla ( \rho ^{-1} \mathbf{B} \cdot \mathbf{u} ) \right)^\flat   
,\\
\widetilde{ \mathbf{B} }&= \mathbf{B} + \theta \left( \frac{\delta h}{\delta\mathbf{m}}\times \operatorname{curl} \frac{\delta C}{\delta \mathbf{B} }- \frac{\delta C}{\delta\mathbf{m}}\times \operatorname{curl} \frac{\delta h}{\delta \mathbf{B} }  - \nabla \phi  \right)^\flat = \mathbf{B} + \theta \left( \mathbf{u} \times \operatorname{curl} \mathbf{u}  - \rho ^{-1} \mathbf{B} \times \operatorname{curl} \mathbf{B} - \nabla \phi  \right) ^\flat
.
\end{align*}
For vectors in $\mathbb{R}^3$, the $\flat$ operation (flat) is the identity, so we may drop it.  
Our formalism allows us to select which among these variables will be  modified. For example, modifying $ \mathbf{m} $ only, while keeping $\rho , \eta , \mathbf{B} $ unchanged, and using the positive bilinear form $ \gamma _\rho ( \mathbf{u} , \mathbf{v} )=\int_ \mathcal{D} \rho \,\mathbf{u} \cdot \mathbf{v}\, \,d^3x $ produces the following modification of the $ \mathbf{u}$-equation
\begin{equation}\label{mod-moteqn-m-only}
\rho (\partial _t \mathbf{u} + \mathbf{u} \cdot \nabla \mathbf{u} )= - \nabla p+ \mathbf{J} \times \mathbf{B} 
-  \theta \pounds_\mathbf{u}\mathsf{X},
\end{equation}
where $ \mathsf{X}= \rho [ \mathbf{u} , \rho ^{-1} \mathbf{B} ]^\flat$.


\begin{enumerate}
\item
In the {\bf Casimir-dissipative case}, we find the decay rate for cross helicity to be
\begin{equation}
\frac{d}{dt} \int_ \mathcal{D} \mathbf{u} \cdot \mathbf{B}\, \,d^3x = - \theta \int_ \mathcal{D} \frac{1}{\rho } | \mathsf{X}| ^2 \,d^3x 
\,.
\label{CasDis-compMHD}
\end{equation}
In the energy-Casimir stability method for compressible MHD \cite{HMRW1985}, the velocity equilibrium condition  is given by 
\begin{equation}
\frac{\delta (h+C)}{\delta \mathbf{m}} = \mathbf{u} +  \rho ^{-1} \mathbf{B} = 0
\quad\hbox{at equilibrium.}
\label{vel-eqmcond1}
\end{equation}
Under this condition, and consistently with Theorem \ref{SteadyStates}, the commutator vanishes in the definition of the 1-form density $ \mathsf{X}= \rho [ \mathbf{u} , \rho ^{-1} \mathbf{B} ]$ when the flow velocity achieves its equilibrium form. That is, $| \mathsf{X}|^2=0$ in \eqref{CasDis-compMHD} for  energy-Casimir equilibria of ideal compressible fluid flow.

\item
We now consider the {\bf energy-dissipative case}. As before, we consider only the modification in the momentum equations.
Using the same positive bilinear form $ \gamma _\rho ( \mathbf{u} , \mathbf{v} )=\int_ \mathcal{D} \rho \,\mathbf{u} \cdot \mathbf{v}\, \,d^3x $ as before, we get, cf. equation \eqref{mod-moteqn-m-only},
\begin{equation}\label{moderg-moteqn-m-only}
\rho (\partial _t \mathbf{u} + \mathbf{u} \cdot \nabla \mathbf{u} )= - \nabla p+ \mathbf{J} \times \mathbf{B} 
+ \theta \pounds_\mathbf{w}\mathsf{X},
\end{equation}
where $ \mathbf{w} = \delta C/ \delta \mathbf{m} =\rho ^{-1} \mathbf{B}$ and $ \mathsf{X}= \rho [ \mathbf{u} , \rho ^{-1} \mathbf{B} ]^\flat$. Finally, one finds the energy decay rate, cf. \eqref{CasDis-compMHD},
\begin{equation}\label{ErgDis-compMHD}
\frac{d}{dt} h( \mathbf{m} , \rho , \eta , \mathbf{B} )=- \theta \int_ \mathcal{D} \frac{1}{\rho } | \mathsf{X}| ^2 \,d^3x 
\,.
\end{equation}

As one expects from Theorem \ref{SteadyStates}, the energy and Casimir decay rates for the cross-helicity case are again the same. 
\end{enumerate}

\section{Conclusions}

This paper has introduced modifications of the equations of ideal fluid dynamics with advected quantities to allow selective decay of either the energy $h$ or the Casimir quantities $C$ in its Lie-Poisson formulation. The quantity selected to be dissipated (energy or Casimir, respectively) is shown to decrease in time until the modified system reaches a { set of equilibrium states that contains the ideal energy-Casimir equilibria satisfying $\delta(h+C)=0$.} The result holds for Lie-Poisson equations in general, independently of the Lie algebra and the choice of Casimir. { The ideal energy-Casimir equilibrium conditions also produce equilibrium states of the selective decay equations. However, in certain situations the time-asymptotic states of the selective decay process satisfies} only some of the energy-Casimir conditions, not all of them. This is explained in the proof of the main result, Theorem \ref{SteadyStates}. 


We interpret the selective decay modifications of the equations as dynamical nonlinear parameterizations of the interactions among different scales, obtained by introducing new nonlinear pathways to dissipation. The modification process is illustrated with a number of selective decay examples that pass to equilibria for magnetohydrodynamics (MHD) in 2D and 3D by decay of either the energy or the Casimirs. 


Equations \eqref{Casimir_dissipation} and  \eqref{LP_form-hdis-mu} provide the constraint forces that will guide the ideal MHD system into a particular class of equilibria, by decreasing, respectively, either the Casimir at constant energy, or vice versa. The existence of a constraint force that will dynamically guide the ideal MHD system into a certain class of equilibria (or preserve it once it has been obtained) may provide useful ideas in the design and control of magnetic confinement devices. 

The Lagrange-d'Alembert variational principle extends Hamilton's principle to the case of forced systems, including nonholonomically constrained systems (\cite{Bl2004}). We have explained following \cite{FGBHo2012} in Section \ref{LdA-princ} how the constraint forces for Casimir-dissipative LP equations \eqref{Casimir_dissipation} and energy-dissipative LP equations \eqref{LP_form-hdis-mu} can be obtained from the Lagrange-d'Alembert principle.

The selective decay approach for fluid flows with advected quantities could also be useful in generating higher-order stable dissipative numerical schemes that generalize the 2D anticipated vorticity method to 3D and include the advected quantities. Thus, the numerical implementation of the present approach may lead to improved numerical schemes that dissipate energy, but conserve Casimirs. This could lead to generalizations for MHD of the anticipated vorticity method for parametrizing subgrid scale barotropic and baroclinic eddies in quasi-geostrophic models, \cite{SaBa1985}. 

In all of the present discussions, we have assumed that the solutions of the modified (dissipative) equations possess long-time existence. However, from the viewpoint of mathematical analysis, the possibility of blow up in finite time for these equations remains an open problem.

\subsection*{Acknowledgments} We are grateful to A. Arnaudon, C. J. Cotter, R. Hide, C. Tronci, G. Vallis and B. A. Wingate for fruitful and thoughtful discussions during the course of this work. We also thank A. Arnaudon for making the figures of the rigid body flows. FGB was partially supported by a ``Projet Incitatif de Recherche'' contract from the Ecole Normale Sup\'erieure de Paris. DDH is grateful for partial support by the European Research Council Advanced Grant 267382 FCCA.
 
\footnotesize

\bibliographystyle{new}

\addcontentsline{toc}{section}{References}

\begin{thebibliography}{300}

%

\bibitem[Arnold(1966)]{Ar1966}
V.I. Arnold [1966], Sur la g\'eom\'etrie differentielle des groupes de Lie de dimension infinie et ses applications \`a l'hydrodynamique des fluides parfaits,  
\textit{Ann. Inst. Fourier}  \textbf{16}, 319--361. 

\bibitem[Arnold(1969)]{Ar1969}
V. I. Arnold [1969], The Hamiltonian character of the Euler equation for the
dynamics of a rigid body and an ideal fluid.
\textit{Usp. Mat. Nauk} \textbf{24}(3), 225--26 (In Russian)

\bibitem[Arnold(1978)]{Ar1978}
V. I. Arnold [1978], \textit{Mathematical Methods of Classical Mechanics}, Springer, N.Y.

\bibitem[Arnold and Khesin(1998)]{ArKh1998}
V. I. Arnold and B. Khesin [1998],
\textit{Topological Methods in Hydrodynamics}, Springer, N.Y.


\bibitem[Bialynicki-Birula and Morrison(1991)]{BBMo1991}
I. Bialynicki-Birula and P. J. Morrison [1991], Quantum mechanics as a generalization of Nambu dynamics to the Weyl-Wigner formalism, \textit{Phys. Lett. A} \textbf{158}, 453--457.  

\bibitem[Bloch(2004)]{Bl2004}
A. M. Bloch (with the collaboration of: J. Baillieul, P. E. Crouch and J. E. Marsden) [2004],
\textit{Nonholonomic Mechanics and Control}, Springer, N.Y.

\bibitem[Bloch et al.(1996)]{BKMR1996}
A. Bloch, P.S. Krishnaprasad, J.E. Marsden, T.S. Ratiu [1996],
The Euler--Poincar\'e equations and double bracket dissipation, 
\textit{Comm. Math. Phys.} \textbf{175}, 1--42.

\bibitem[Bloch et al.(2012)]{BlMoRa2012}
A. Bloch, P. J. Morrison and T. S. Ratiu [2012], 
Gradient flows in the normal and K\"ahler metrics and triple bracket
generated metriplectic systems,
\url{http://arxiv.org/abs/1208.6193v1}

\bibitem[Bou-Rabee, Marsden, and Romero(2008)]{Bou2008}
N. M. Bou-Rabee and J. E. Marsden and L. A. Romero [2008]  Dissipation-induced heteroclinic orbits in tippe tops,
\textit{SIAM review}, \textbf{50}, 325--344


\bibitem[Brockett(1991)]{Br1991}
R. W. Brockett [1991], Dynamical systems that sort lists, diagonalise matrices, and solve linear programming problems, \textit{Linear Algebr. Appl.}, \textbf{146}, 79--91.

\bibitem[Brody, Ellis and Holm(2008)]{BrElHo2008}
D. C. Brody, D. C. P. Ellis and D. D. Holm [2008],
Hamiltonian statistical mechanics.
\textit{J. Phys. A: Math. Theor.} {\bf 41}, 502002.

\bibitem[Brown, Canfield and Pertsoy(1999)]{BrCaPe1999}
M. R. Brown, R. C. Canfield and A. A. Pertsoy [1999],
\textit{Magnetic Helicity in Space and Laboratory Plasmas},
Geophysical monograph series vol 11,
American Geophysical Union, Washington, DC.

\bibitem[Chandrasekhar(1961)]{Chandra1961}
S. Chandrasekhar [1961], \textit{Hydrodynamic and Hydromagnetic Instabilities} (Oxford Univ Press, London, New York).

\bibitem[Frisch, Pouquet, Leorat and Mazure(1975)]{FPLM1975}
U. Frisch, A. Pouquet, J. Leorat and A. Mazure [1975], Possibility of an
inverse cascade of magnetic helicity in magnetohydrodynamic turbulence. 
\textit{J. Fluid Mech.} \textbf{68}, 769-778. 

%
%
%
%

\bibitem[Gay-Balmaz and Holm(2013)]{FGBHo2012}
F. Gay-Balmaz and D. D. Holm [2013],
Parameterizing interaction of disparate scales: Selective decay by Casimir dissipation in fluids. 
{\it Nonlinearity} 26(2):495--524. 
 \url{http://arxiv.org/abs/1206.2607}


\bibitem[Gay-Balmaz and Ratiu(2009)]{FGBRa2009}
F. Gay-Balmaz and T. S. Ratiu  [2009],
The geometric structure of complex fluids,
{\it Adv. Appl. Math.} {\bf 42}, 176--275
%
%

\bibitem[Greene and Karlson(1969)]{GrKa1969} 
J. M. Greene and E. T. Karlson [1969],
Variational principle for stationary magnetohydrodynamic equilibria, 
{\it Phys. Fluids} {\bf 12}, 561--567.

\bibitem[Grmela(1984)]{Gr1984}
M. Grmela [1984], Bracket formulation of dissipative fluid mechanics equations,
\textit{Phys. Lett. A}, \textbf{102}, 355--358.
%
%

\bibitem[Holm(1987)]{Ho1987}
D. D. Holm [1987], 
Hall magnetohydrodynamics: conservation laws and Lyapunov stability, 
{\it Phys. Fluids},  {\bf35}, 1310--1322.
%

\bibitem[Holm(2011)]{Ho2011}
D. D. Holm [2011], {\it Geometric mechanics. Part I: Dynamics and symmetry; Part II: Rotating, Translating and Rolling.}
Imperial College Press, London; distributed by World Scientific Publishing Co., Hackensack, N. J.

\bibitem[Holm and Kupershmidt(1983a)]{HoKu1983a} 
D. D. Holm and B. A. Kupershmidt [1983a],
Poisson brackets and Clebsch representations 
for magnetohydrodynamics, multifluid  plasmas, and elasticity, 
{\it Physica D} {\bf 6} 347--363.

\bibitem[Holm and Kupershmidt(1983b)]{HoKu1983b} 
D. D. Holm and B. A. Kupershmidt [1983b],
Noncanonical Hamiltonian formulation 
of ideal  magnetohydrodynamics, 
{\it Physica D} {\bf 7} 330--333.

\bibitem[Holm, Marsden, and Ratiu(1998)]{HMR1998}
D. D. Holm, J. E. Marsden and T. S. Ratiu [1998], 
The Euler--Poincar\'e equations and semidirect products.
{\it Adv. in Math.}, {\bf 137} (1998) 1--81.
Preprint at \url{http://xxx.lanl.gov/abs/chao-dyn/9801015}.

\bibitem[Holm et al.  (1984)]{HMRW1984} 
D. D. Holm, J. E. Marsden, T. Ratiu and  A. Weinstein [1984], 
Stability of rigid body motion using the Energy-Casimir method.
In: \textit{Fluids and Plasmas: geometry and dynamics}. Contemporary Mathematics \textbf{28}, American Mathematical Society , pp. 15-23. 
\url{http://resolver.caltech.edu/CaltechAUTHORS:20100812-150925524}

\bibitem[Holm et al.(1985)]{HMRW1985} 
D. D. Holm, J. E. Marsden, T. Ratiu and  A. Weinstein [1985], 
Nonlinear stability of fluid and plasma equilibria, 
{\it Physics Reports} {\bf 123} (1985) 1--116.

\bibitem[Holm, Putkaradze and Tronci(2008)]{HoPuTr2008}
D. D. Holm, V. Putkaradze and C. Tronci [2008], 
Geometric gradient-flow dynamics with singular solutions,
\textit{Physica D}, \textbf{237}, 2952--2965 

\bibitem[Holm and Tronci(2012)]{HoTr2012}
D. D. Holm and C. Tronci [2012], 
Multiscale turbulence models based on convected fluid microstructure.
Preprint at \url{http://arxiv.org/pdf/1203.4545.pdf}

\bibitem[Kandrup(1991)]{Ka1991}
H.E. Kandrup [1991], The secular instability of axisymmetric collisionless star cluster. \textit{Astrophy. J.} \textbf{380}, 511--514.

\bibitem[Kaufman(1984)]{Ka1984}
A. N. Kaufman [1984], Dissipative Hamiltonian systems: A unifying principle,
\textit{Phys. Lett. A}, \textbf{100}, 419--422.

\bibitem[Kawazura and Hameiri(2012)]{KaHa2012}
Y. Kawazura and E. Hameiri  [2012],
The complete set of Casimirs in Hall-magnetohydrodynamics
\textit{Phys. of Plasmas} \textbf{19} 082513.

\bibitem[Kraichnan(1967)]{Kr1967}
R. H. Kraichnan [1967], Inertial ranges in two-dimensional turbulence. 
\textit{Phys. Fluids} \textbf{10}, 1417--1423.

\bibitem[Kraichnan(1971)]{Kr1971}
R. H. Kraichnan [1971], Inertial-range transfer in two- and three-dimensional turbulence.  \textit{J. Fluid Mech.} \textbf{47}, 525--535.

\bibitem[Kraichnan and Montgomery(1980)]{KrMo1980}
R. H. Kraichnan and D. Montgomery [1980], Two-dimensional turbulence. \textit{Rep. Prog. Phys.} \textbf{43}, 547--619.

\bibitem[Landau and Lifshitz(1935)]{LL1935}
L. Landau and E. Lifshitz, [1935], \textit{Phys. Zeitsch. der Sow.} \textbf{8}, 153, 
reprint [2008] \textit{Ukr. J. Phys.} \textbf{53}, 14.

\bibitem[Levy, Dubrulle and Chavanis(2006)]{LeDuCh2006}
N. Leprovost, B. Dubrulle and P.-H. Chavanis [2006], Dynamics and thermodynamics of axisymmetric flows: Theory.
\textit{Phys. Rev. E} \textbf{73}, 046308.

\bibitem[Marsden and Ratiu(1994)]{MaRa1994}
J. E. Marsden  and T. S. Ratiu [1994], {\em Introduction to Mechanics and Symmetry}. 
Texts in Applied Mathematics, Vol. 75. New York: Springer.

\bibitem[Marsden, Ratiu and Weinstein(1984)]{MaRaWe1984}
J.~E. Marsden, T.~S. Ratiu and A. Weinstein [1984], Semidirect product and
reduction in mechanics, \textit{Trans. Amer. Math. Soc.},
\textbf{281}, 147--177.

\bibitem[Marshall and Adcroft(2010)]{MaAd2010}
D. P. Marshall and A. J. Adcroft [2010], 
Parameterization of ocean eddies: Potential vorticity mixing, energetics and ArnoldÕs first stability theorem.
\textit{Ocean Modelling} \textbf{32}, 188--204.

\bibitem[Matthaeus and Montgomery(1980)]{MaMo1980}
W. H. Matthaeus and D. Montgomery [1980],
Selective decay hypothesis at high mechanical and magnetic Reynolds numbers, 
\textit{Ann. N. Y. Acad. Sci.}, \textbf{357}, 203--222.

\bibitem[Mininni, Pouquet and Sullivan(2008)]{MiPoSu2008}
P. Mininni, A. Pouquet and P. Sullivan [2008],
Two examples from geophysical and astrophysical turbulence on modeling
disparate scale interactions.
Appeared in 
{\it Computational Methods for the Atmosphere and the Ocean} (a special volume of the {\it Handbook of Numerical Analysis}), 
edited by R. Temam and J. Tribbia. pp. 333--377.
Academic Press Professional, Inc.  San Diego, CA, USA.

\bibitem[Morrison(1984)]{Mo1984}
P. J. Morrison [1984], Bracket formulation for irreversible classical fields, 
\textit{Phys. Lett. A}, \textbf{100}, 423--427.

\bibitem[Morrison(1986)]{Mo1986}
P. J. Morrison [1986], A paradigm for joined Hamiltonian and dissipative systems, \textit{Physica D}, \textbf{18}, 410--419.

\bibitem[Morrison and Greene (1980)]{MoGr1980}
P. J. Morrison and J. M. Greene [1980], 
Noncanonical Hamiltonian density formulation of hydrodynamics and ideal magnetohydrodynamics
{\it Phys. Rev. Lett.} {\bf 45}, 790--794. Errata [1982] {\it Phys. Rev. Lett.} {\bf 48}, 569.

\bibitem[Montgomery and Bates(1999)]{MoBa1999}
D. C. Montgomery and J. W. Bates [1999],
Helicity and its role in the varieties of magnetohydrodynamic turbulence.
In \cite{BrCaPe1999} pp 33--46. 

\bibitem[Nambu(1973)]{Nambu1973} 
Y. Nambu [1973],
Generalized Hamiltonian mechanics, 
\textit{Phys. Rev. D} {\bf7}, 2405--2412.

\bibitem[Olver(2000)]{Ol2000}
P. J. Olver [2000],
\textit{Applications of Lie groups to Differential Equations},
Springer, N. Y.

\bibitem[\"Ottinger(2005)]{Ot2005}
H. C. \"Ottinger [2005], \textit{Beyond Equilibrium Thermodynamics}, Wiley-Interscience, N. Y.

\bibitem[Sadourny and Basdevant(1981)]{SaBa1981}
R. Sadourny and C. Basdevant [1981], Une classe d'op\'erateurs adapt\'es \`a
la mod\'elisation de la diffusion turbulente en dimension deux,
\textit{C. R. Acad. Sci. Paris}, \textbf{292}, 1061--1064.

\bibitem[Sadourny and Basdevant(1985)]{SaBa1985}
R. Sadourny and C. Basdevant [1985], Parametrization of subgrid scale barotropic and baroclinic eddies in quasi-geostrophic models: anticipated potential vorticity method,  \textit{J. Atm. Sci.}, \textbf{42}(13), 1353--1363.

\bibitem[Salmon(2005)]{Sa2005}
R. Salmon [2005], A general method for conserving quantities related to
potential vorticity in numerical models,
\textit{Nonlinearity} \textbf{18} R1ÐR16.

\bibitem[Shepherd(1990)]{Shep1990}
T. G. Shepherd [1990],
A general method for finding extremal states of 
Hamiltonian dynamical systems, with applications 
to perfect fluids 
\textit{J. Fluid Mech.} \textbf{213}, 573--587.

\bibitem[Taylor(1974)]{Ta1974}
J. B. Taylor [1974], Relaxation of toroidal plasma and generation of reverse magnetic fields, 
\textit{Phys. Rev. Lett.} \textbf{33}, 1139-1143. 

\bibitem[Taylor(1986)]{Ta1986}
J. B. Taylor [1986], Relaxation and magnetic reconnection in plasmas, 
\textit{Rev. Mod. Phys.} \textbf{58}, 741--763. 

\bibitem[Vallis, Carnevale, and Young(1989)]{VaCaYo1989}
G.~K. Vallis, G.~F. Carnevale and W.~R. Young [1989], 
Extremal energy properties and construction of stable solutions of the Euler equations, 
\textit{J. Fluid Mech.} \textbf{207}, 133--152.

\bibitem[Vallis and Hua(1988)]{VaHu1988}
G.~K. Vallis and B.-L. Hua [1988], Eddy viscosity of the anticipated potential vorticity method, 
\textit{J. Atm. Sci.}, \textbf{45}(4), 617--627.

\bibitem[Woltjer(1958)]{Wo1958}
L. Woltjer [1958], A theorem on force-free magnetic fields, \textit{Proc. Nat. Acad. Sci. USA}, \textbf{44}, 489--491.

\bibitem[Woltjer(1959)]{Wo1959}
L. Woltjer [1959], Hydromagnetic equilibrium. III. Axisymmetric incompressible media, \textit{Astrophys. J.} \textbf{130}, 400--404.

\bibitem[Woltjer(1960)]{Wo1960}
L. Woltjer [1960], On the theory of hydromagnetic equilibrium \textit{Rev. Mod. Phys.} \textbf{32}, 914--915.

\end{thebibliography}

\end{document}